\documentclass[11pt]{article}
\usepackage{authblk}
\usepackage{mathrsfs}
\usepackage{hyperref}
\usepackage{enumerate}
\usepackage{algorithm}
\usepackage{algpseudocode}
\usepackage{epsfig}
\usepackage{wrapfig}
\usepackage{amsmath}
\usepackage{amssymb}
\usepackage{latexsym}
\usepackage{amsfonts}
\usepackage{amsthm}
\usepackage[numbers,sort&compress]{natbib}
\usepackage{graphicx}
\usepackage{color}
\usepackage{fullpage}

\newcommand {\OPT } {\mathrm{OPT}}
\newcommand {\inter} {\mathrm{int}}

\newtheorem{definition}{Definition}

\newtheorem{lemma}{Lemma}
\newtheorem{theorem}[lemma]{Theorem}

\newtheorem{observation}{Observation}

\newtheorem{corollary}[lemma]{Corollary}
\newtheorem{claim}[lemma]{Claim}

\begin{document}
\newif\ifFull 
\Fulltrue

\title{
{A PTAS for Three-Edge Connectivity in Planar Graphs}
}
\author{Glencora Borradaile}
\author{Baigong Zheng }
\affil{Oregon State University \\
\{glencora, zhengb\}@eecs.oregonstate.edu
}

\date{}
\maketitle
\thispagestyle{empty}
\begin{abstract}
We consider the problem of finding the minimum-weight subgraph that satisfies given connectivity requirements. Specifically, given a requirement $r \in \{0,1,2,3\}$ for every vertex, we seek the minimum-weight subgraph that contains, for every pair of vertices $u$ and $v$, at least $\min\{ r(v), r(u)\}$ edge-disjoint $u$-to-$v$ paths. We give a polynomial-time approximation scheme (PTAS) for this problem when the input graph is planar and the subgraph may use multiple copies of any given edge. This generalizes an earlier result  for $r \in \{0,1,2\}$. In order to achieve this PTAS, we prove some properties of triconnected planar graphs that may be of independent interest. 
\end{abstract}
\vfill
{\small 
\ifFull
\else
\paragraph{Acknowledgements}
We thank Hung Le, Amir Nayyeri and David Pritchard for helpful discussions.
\fi
This material is based upon work supported by the National Science Foundation under Grant No.\ CCF-1252833.}

\newpage
\setcounter{page}{1}
\section{Introduction}


The survivable network design problem aims to find a low-weight subgraph that connects a subset of vertices and will remain connected despite edge failures, an important requirement in the  field of telecommunications network design. This problem can be formalized as the $I$-edge connectivity problem for an integer set $I$ as follows:  for an edge-weighted graph $G$ with a requirement function on its vertices $r:V(G)\to I$, we say a subgraph $H$ is a feasible solution if for any pair of vertices $u, v\in V(G)$, $H$ contains $\min\{r(u), r(v)\}$ edge-disjoint $u$-to-$v$ paths; the goal is to find the cheapest such subgraph.  
In the {\em relaxed} version of the problem, $H$ may contain multiple (up to $\max I$) copies of $G$'s edges ($H$ is a {\it multi}-subgraph) in order to achieve the desired connectivity, paying for the copies according to their multiplicity;  otherwise we refer to the problem as the  {\it strict} version. Thus $I = \{1\}$ corresponds to the minimum spanning tree problem and $I = \{0,1\}$ corresponds to the minimum Steiner tree problem.  Here our focus is when $\max I \ge 2$.

This problem and variants have a long history. 
The $I$-edge connectivity problem, except when $I = \{1\}$, is MAX-SNP-hard~\cite{CL99}. There are constant-factor approximation algorithms for the strict $\{2\}$-edge-connectivity problem, with the best factor being 5/4 due to Jothi, Raghavachari and Varadarajan~\cite{FJ81, KV94, JRV03}. Ravi~\cite{Ravi92} gave a 3-approximation for the
strict $\{0,2\}$-edge-connectivity problem and Klein and Ravi~\cite{KR93} gave a 2-approximation for the strict $\{0,1,2\}$-edge-connectivity problem. For general requirement, Jain~\cite{Jain01} gave a 2-approximation for both the strict version and the relaxed version of the problem.

We study this problem in planar graphs.  In planar graphs, the $I$-edge connectivity problem, except when $I = \{1\}$, is NP-hard (by a reduction from Hamiltonian cycle).  Berger, Czumaj, Grigni, and Zhao~\cite{BCGZ05} gave a polynomial-time approximation scheme\footnote{A polynomial-time approximation scheme for an minimization problem is an algorithm that, given a fixed constant $\epsilon > 0$, runs in polynomial time and returns a solution within $1 +\epsilon$ of optimal. The algorithm's running time need not be polynomial in $\epsilon$.} (PTAS) for the relaxed $\{1,2\}$-edge-connectivity problem, and Berger and Grigni~\cite{BG07} gave a PTAS for the strict $\{2\}$-edge-connectivity problem.  Borradaile and Klein~\cite{BK08} gave an {\em efficient}\footnote{A PTAS is efficient if the running time is bounded by a polynomial whose degree is independent of $\epsilon$.} PTAS (EPTAS) for the relaxed $\{0,1,2\}$-edge-connectivity problem\footnote{Note that in Borradaile and Klein~\cite{BK08} claimed their PTAS would generalize to relaxed $\{0,1,\ldots,k\}$-edge-connectivity, but this did not come to fruition.}.  The only planar-specific algorithm for non-spanning, strict edge-connectivity is a PTAS for the following problem: given a subset $R$ of edges, find a minimum weight subset $S$ of edges, such that for every edge in $R$, its endpoints are two-edge-connected in $R \cup S$~\cite{KMZ15}; otherwise, the best known results for the strict versions of the edge-connectivity problem when $I$ contains 0 and 2 are the constant-factor approximations known for general graphs.

In this paper, we give an EPTAS for the relaxed $\{0,1,2,3\}$-edge-connectivity problem in planar graphs.  This is the first PTAS for connectivity beyond 2-connectivity in planar graphs:
\begin{theorem}
For any $\epsilon>0$ and any planar graph instance of the relaxed $\{0,1,2,3\}$-edge connectivity problem, there is an $O(n\log n)$-time algorithm that finds a solution whose weight is at most $1+\epsilon$ times the weight of an optimal solution.
\end{theorem}

In order to give this EPTAS, we must prove some properties of triconnected (three-vertex connected) planar graphs that may be of independent interest. 
One simple-to-state corollary of the sequel is:
\begin{theorem}\label{thm: twoterminals}
In a planar graph that minimally pairwise triconnects a set of terminal vertices, every cycle contains at least two terminals.
\end{theorem}

In the remainder of this introduction we overview the PTAS for network design problems in planar graphs~\cite{BK08} that we use for the relaxed $\{0,1,2,3\}$-edge connectivity problem.  In this overview we highlight the technical challenges that arise from handling 3-edge connectivity.  We then overview why we use properties of vertex connectivity to address an edge connectivity problem and state our specific observations about triconnected planar graphs that we require for the PTAS framework to apply.
In the remainder, 2-ECP refers to ``the relaxed $\{0,1,2\}$-edge-connectivity problem" and 3-ECP refers to ``the relaxed $\{0,1,2,3\}$-edge-connectivity problem".

\subsection{Overview of PTAS for 2-ECP}

The planar PTAS framework grew out of a PTAS for travelling salesperson problem~\cite{Klein08} and has been used to give PTASes for Steiner tree~\cite{BKK07, BKM09}, Steiner forest~\cite{BHM11} and 2-EC~\cite{BK08} problems.   
For simplicity of presentation, we follow the PTAS whose running time is doubly exponential in $1/\epsilon$~\cite{BKK07}; this can be improved to singly exponential as for Steiner tree~\cite{BKM09}.  
Note that for all these  problems (except Steiner forest, which requires a preprocessing step to the framework), the optimal value of the solution is within a constant factor of the optimal value of a Steiner tree on the same terminal set where we refer to vertices with non-zero requirement as  {\em terminals}.  
Let OPT be the weight of an optimal solution to the given problem. 

The PTAS for 2-ECP relies on an algorithm to find the {\em mortar graph} $MG$ of the input graph $G$. 
The mortar graph is a grid-like subgraph of $G$ that spans all the terminals and has total weight bounded by $f(\epsilon)$ times the minimum weight of a Steiner tree spanning all the terminals. 
To construct the mortar graph, we can first find an approximate Steiner tree for all terminals and then recursively augment it with some short paths.
For each face of the mortar graph, the subgraph of $G$ that is enclosed by that face (including the boundary of the face) is called a {\em brick}.
It is shown that there is a nearly optimal solution for 2-ECP whose intersection with each brick is a set of non-crossing trees~\cite{BK08}. 
Further, it is proved that each such tree has only $g(\epsilon)$ leaves and its leaves are a subset of a designated vertex set, called {\em portals}, on the boundary, allowing these trees to be computed efficiently~\cite{EMV87}.

In the following, $O$-notation hides factors depending on $\epsilon$ that only affect the running time.  
The PTAS for 2-ECP consists of the following steps:

\begin{description}
\item [Step 1:] Find a subgraph that satisfies two properties: its weight is at most $O(\OPT)$ and it contains a $(1+\epsilon)$-approximate solution. Such a subgraph is often referred to as a {\em spanner} in the literature. It is sufficient to solve the problem in the spanner.
\begin{description}
\item [(1)] Find the mortar graph $MG$.
\item [(2)] For each brick of $MG$, designate as {\em portals} a constant number (depending on $\epsilon$) of vertices on the boundary of each brick. 
\item [(3)] Find Steiner trees for each subset of portals in each brick by the algorithm of Erickson et al.~\cite{EMV87}. All the Steiner trees of each brick together with mortar graph $MG$ form the spanner.
\end{description}
\item [Step 2:] Decompose the spanner into a set of subgraphs, called {\em slices}, that satisfy the following properties.
\begin{description}
\item [(1)] Each slice has bounded {\em branchwidth}.
\item [(2)] Any two slices share at most one simple cycle.
\item [(3)] Each edge can belong to at most two slices; such edges form the boundaries of slices.
\item [(4)] The weight of all edges that belong to two slices is at most  $\epsilon \OPT$. \end{description}
\item [Step 3:] Select a set of ``artificial" terminals with connectivity requirements on the boundaries of slices from the previous step to achieve the following:
\begin{itemize}
\item For each slice, there is a feasible solution with respect to original and artificial terminals whose weight is bounded by the weight of the slice's boundary plus the weight of the intersection of an optimal solution with the slice.
\item The union over all slices of such feasible solutions is a feasible solution for the original graph.
\end{itemize}
\item [Step 4:] Solve the 2-ECP with respect to original and artificial terminals in each slice by dynamic programming.
\item [Step 5:] 
Convert the optimal solutions from the previous step to a solution for the spanner.
\end{description}

We can construct the spanner in $O(n \log n)$ time~\cite{BKM09}. 
We can identify the boundary edges in Step 2 by doing breadth-first search in the planar dual and then applying the shifting technique of Baker~\cite{Baker94}, which can be done in linear time. 
With these edges we can decompose the spanner into slices in linear time.
Step 3 can be done in linear time since the slices form a tree structure and we only need to choose as a new terminal any vertex on the boundary cycle if the cycle separates any two original terminals. If a boundary cycle separates two terminals requiring two-edge connectivity, the connectivity requirement of the new terminal on that cycle is 2, otherwise 1.
By standard dynamic programming techniques we can solve the 2-ECP in the graph with bounded branchwidth in linear time.
Step 5 can be done in linear time.
So the total running time of the PTAS for 2-ECP is $O(n\log n)$.

We will generalize this PTAS to 3-ECP.  
The differences for 3-ECP are Step 3 and Step 4.
For Step 3, we set the connectivity requirement of the new terminal to 3 if its corresponding boundary cycle separates two terminals requiring three-edge connectivity.
For Step 4, we use the dynamic programming for $k$-ECP on graphs with bounded branchwidth\ifFull given in Section~\ref{sec:dp}, \fi, which is \ifFull inspired by \else similar to \fi that for the $k$-vertex-connectivity spanning subgraph problem in Euclidean space given by Czumaj and Lingas in \cite{CL98,CL99}.

To prove this PTAS is correct, we need to show the subgraph obtained from Step 1 
is a spanner for 3-ECP, which is the challenge of this work (as with most applications of the PTAS framework). 
For any fixed $0 < \epsilon < 1$ and input graph $G$, a subgraph $H$ of $G$ is a spanner for 3-ECP, if it has the following two properties.
\begin{description}
\item [(1)] The weight of $H$ is at most $O(\OPT)$.
\item [(2)] There is a $(1+\epsilon)$-approximate solution using only the edges of $H$.
\end{description}
The weight of the spanner found in Step 1 is at most $f(\epsilon)$ times the minimum weight of a Steiner tree spanning all the terminals. 
Since $\OPT$ is more than the minimum weight of a Steiner tree, the weight of our spanner is at most $O(\OPT)$.
The second property will be guaranteed by the following Structure Theorem, which is the main focus of this paper.
\begin{theorem}[Structure Theorem]\label{thm:struct}
For any $\epsilon>0$ and any planar graph instance $(G,w, r)$ of 3-ECP, there exists a feasible solution $S$ in our spanner such that 
\begin{itemize}
\item the weight of $S$ is at most $(1+c\epsilon) \OPT$ where $c$ is an absolute constant, and 
\item the intersection of $S$ with the interior of any brick is a set of trees whose leaves are on the boundary of the brick and each tree has a number of leaves depending only on $\epsilon$.
\end{itemize}
\end{theorem}
The {\em interior} of a brick is all the edges of a brick that are not on the boundary.
We denote the interior of a brick $B$ by $\inter(B)$.
Consider a brick $B$ of $G$ whose boundary is a face of $\mathrm{MG}$ and consider the intersection of $\OPT$ with the interior of this brick, $\OPT \cap \inter(B)$.  
To prove the Structure Theorem, we will show that: 
\begin{description}
\item [P1:] $\OPT \cap \inter(B)$ can be partitioned into a set of trees $\mathcal T$ whose leaves are on the boundary of $B$.
\item [P2:] If we replace any tree in $\mathcal T$ with another tree spanning the same leaves, the result is a feasible solution.
\item [P3:] There is another set of $O(1)$ trees $\mathcal T'$ that costs at most a $1+\epsilon$ factor more than $\mathcal T$, such that each tree of $\mathcal T'$ has $O(1)$ leaves and $(\OPT \setminus {\mathcal T} ) \cup {\mathcal T}'$ is a feasible solution.\footnote{Strictly, we also add multiple copies of edges from the boundary of $B$ to guarantee feasibility of $(\OPT\setminus {\mathcal T} ) \cup {\mathcal T}'$.}
\end{description}

Property P1 implies that we can decompose an optimal solution into a set of trees inside of bricks.
Property P2 shows that we can treat those trees independently with regard to connectivity, and this gives us hope that we can replace $\OPT \cap \inter(B)$ with some Steiner trees with terminals on the boundary which we can efficiently compute in planar graphs~\cite{EMV87}.
Property P3 shows that we can compute an approximation to $\OPT \cap \inter(B)$ by guessing $O(1)$ leaves.  
Those approximations can be combined efficiently in the remaining steps of the PTAS.

For the Steiner tree problem, P1 and P2 are nearly trivial to argue; the bulk of the work is in showing P3~\cite{BKK07}.  

For 2-ECP, P1 depends on first converting $G$ and $\OPT$ into $G'$ and $\OPT'$ such that $\OPT'$ biconnects  (two-vertex connects) the terminals requiring two-edge connectivity and using the relatively easy-to-argue fact that every cycle of $\OPT'$ contains at least one terminal. By this fact, a cycle in $\OPT'$ must contain a vertex of the brick's boundary (since $MG$ spans the terminals), allowing the partition of $\OPT' \cap \inter(B)$ into trees.  P2 and P3 then require that two-connectivity across the brick is maintained.

For 3-ECP, P1 is quite involved to show, but further to that, showing Property P2 is also involved;  the issues\footnote{The issues also appear in 2-ECP, but we explain why it is easy to handle in 2-ECP in the next subsection.} are illustrated in Figure~\ref{fig: replace} and  are the focus of Sections~\ref{sec:vert-conn} and~\ref{sec:conn-sep}.  
As with 2-ECP, we convert $\OPT$ into a vertex connected graph to simplify the arguments.  
Given Properties P1 and P2, we illustrate Property P3 by following a similar argument as for 2-ECP; since this requires reviewing more details of the PTAS framework, we cover this in \ifFull Section~\ref{sec:spanner}.\else the full appended version.  
\fi

\begin{figure}[ht]
  \centering
  \includegraphics[scale=1.0]{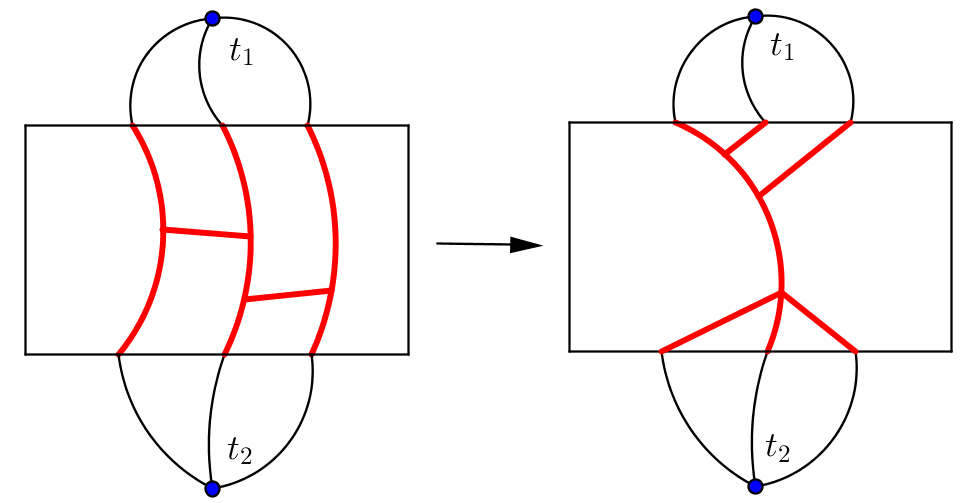}
  \caption{If the bold red tree (left) is $\OPT \cap \inter(B)$ (where $B$ is denoted by the rectangle), replacing the tree with another tree spanning the same leaves (right) could destroy 3-connectivity between $t_1$ and $t_2$.  We will show that such a tree cannot exist in a minimally connected graph.}
  \label{fig: replace}
\end{figure}

\paragraph*{Non-planar graphs}
We point out that, while previously-studied problems that admit PTASes in planar graphs (e.g.\ independent set and vertex cover~\cite{Baker94}, TSP~\cite{Klein08,Klein06,AGKKW98}, Steiner tree~\cite{BKM09} and forest~\cite{BHM11}, 2-ECP~\cite{BK08}) generalize to surfaces of bounded genus~\cite{BDT12}, the method presented in this paper 3-ECP cannot be generalized to higher genus surfaces.  In the generalization to bounded genus surfaces, the graph is preprocessed (by removing some provably unnecessary edges) so that one can compute a mortar graph whose faces bound disks.  This guarantees that even though the input graph is not planar, the bricks are; this is sufficient for proving above-numbered properties in the case of TSP, Steiner tree and forest and 2-ECP.  However, for 3-ECP, in order to prove P2, we require {\em global} planarity, not just planarity of the brick.
To the authors' knowledge, this is the only problem that we know to admit a PTAS in planar graphs that does not naturally generalize to toroidal graphs.

\subsection{Reduction to vertex connectivity}
Here we give an overview of how we use vertex connectivity to argue about the structural properties of edge-connectivity required for the spanner properties.  
\ifFull
First some definitions.
Vertices $x$ and $y$ are $k$-vertex-connected in a graph $G$ if $G$ contains $k$ pairwise vertex disjoint $x$-to-$y$ paths. If $k=3$ ($k=2$), then $x$ and $y$ are also called triconnected (biconnected).
\else
Two vertices are {\em triconnected} (biconnected) if there are three (two) intenally vertex-disjoint paths between them.
\fi
For a subset $Q$ of vertices in $G$ and a requirement function $r:Q\to\{2,3\}$, subgraph $H$ is said to be $(Q,r)$-vertex-connected if every pair of $x$ and $y$ in $Q$ is $k$-vertex-connected for $k=$min$\{r(x),r(y)\}$. We call vertices of $Q$ {\em terminals}.  
If $r(x)=3$ ($r(x)=2$) for all $x\in Q$, we say $H$ is $Q$-triconnected ($Q$-biconnected).  
We say a $(Q,r)$-vertex-connected graph is {\em minimal}, if deleting its any edge or vertex violates the connectivity requirement. 

\begin{figure}[ht]
  \centering
  \includegraphics[scale=1]{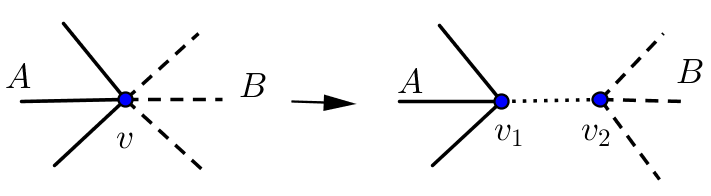}
  \caption{Vertex $v$ is cleaved into vertices $v_1$ and $v_2$. The edges incident to $v$ are partitioned into two sets $A$ and $B$ to become incident to distinct copies.}
  \label{fig: cleave}
\end{figure}
We {\em cleave} vertices to create a subgraph $\OPT'$ of graph $G'$ that is a vertex-connected version of the $\{0,1,2,3\}$-edge-connected multisubgraph $\OPT$ of $G$.  Informally, cleaving a vertex means splitting the vertex into two copies and adding a zero-weight edge between the copies; incident edges choose between the copies in a planarity-preserving way (Figure~\ref{fig: cleave}).  We \ifFull show in Section~\ref{sec:cleave} how to \else can \fi cleave the vertices of $\OPT$ so that if two terminals are $k$-edge-connected in $\OPT$, there are corresponding terminals in $\OPT'$ that are $k$-vertex-connected.  We will argue that $\OPT'$ satisfies Properties P1 and P2 and 
these two properties also hold for OPT,
since OPT$'$ is obtained from OPT by cleavings.

To prove that $\OPT'$ satisfies Property P1, we show that every cycle in OPT$'$ contains at least one terminal (Section~\ref{sec:vert-conn}). To prove that $\OPT'$ satisfies Property P2, we define the notion of a {\em terminal-bounded component}: a connected subgraph is a terminal-bounded component if it is an edge between two terminals or obtained from a maximal terminal-free subgraph $S$ (a subgraph containing no terminals), by adding edges from $S$ to its neighbors (which are all terminals by maximality of $S$). We prove that in a minimal $Q$-triconnected graph any terminal-bounded component is a tree whose leaves are terminals and the following Connectivity Separation Theorem in Section~\ref{sec:conn-sep}:

\begin{theorem}[Connectivity Separation Theorem]\label{thm: tri_disjoint}
Given a minimal $(Q,r)$-vertex-connected planar graph, for any pair of terminals $x$ and $y$ that require triconnectivity (biconnectivity), there are three (two) vertex disjoint paths from $x$ to $y$ in $G$ such that any two of them do not contain edges of the same terminal-bounded tree.
\end{theorem}

\begin{corollary}\label{cor: terminal}
Given a minimal $(Q,r)$-vertex-connected planar graph, for any pair of terminals $x$ and $y$ that require triconnectivity (biconnectivity), there exist three (two) vertex disjoint $x$-to-$y$ paths such that any path that connects any two of those $x$-to-$y$ paths contains a terminal.
\end{corollary}

\noindent This can be viewed as a generalization of the following by Borradaile and Klein for 2-ECP~\cite{BK13}: 

\begin{theorem}\label{thm: ntforbi} {\rm (Theorem 2.8 \cite{BK13}).} Given a graph that minimally biconnects a set of terminals, for any pair of terminals $x$ and $y$ and for any two vertex disjoint $x$-to-$y$ paths, any path that connects these paths must contain a terminal.
\end{theorem}
Note that Theorem~\ref{thm: ntforbi} holds for general graphs while Corollary~\ref{cor: terminal} only holds for planar graphs.  This again underscores why our PTAS does not generalize to higher-genus graphs. Further, Theorem~\ref{thm: ntforbi} implies that for biconnectivity any pair of disjoint $x$-to-$y$ paths has the stated property, while for triconnectivity there can be a pair of $x$-to-$y$ paths that does not have the property (See Figure~\ref{fig: counterex}). So Corollary~\ref{cor: terminal} is the best that we can hope, since it shows there exists a set of disjoint paths that has the stated property;  higher connectivity comes at a price.

\begin{figure}[ht]
\centering
\includegraphics[scale=1.7]{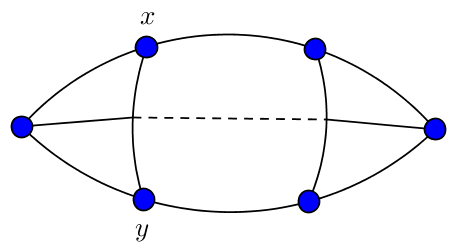}
\caption{A minimal $Q$-triconnected graph. The bold vertices are terminals. The dashed path connects two $x$-to-$y$ paths but it does not contain any terminal.}\label{fig: counterex}
\end{figure}


For OPT$'$, Corollary~\ref{cor: terminal} implies Property P2.
Consider the set of disjoint paths guaranteed by Corollary~\ref{cor: terminal}. If any tree replacement in a brick merges any two disjoint paths, say $P_1$ and $P_2$, in the set (the replacement in Figure~\ref{fig: replace} merges three paths), then the replaced tree must contain at least one vertex of $P_1$ and one vertex of $P_2$.
This implies the replaced tree contains a $P_1$-to-$P_2$ path $P$ such that each vertex in $P$ has degree at least two in the replaced tree. 
Further, $P$ contains a terminal by Corollary~\ref{cor: terminal}. However, all the terminals are in mortar graph, which forms the boundaries of the bricks. 
So $P$ must have a common vertex with the boundary of the brick.  
By Property P1, the replaced tree, which is in the intersection of $\OPT'$ with the interior of the brick, can only contain leaves on the boundary of the brick. 
Therefore, the replaced tree can not contain such a $P_1$-to-$P_2$ path, otherwise there is a vertex in $P$ that has degree one in the tree.
\ifFull We give the complete proof in Section~\ref{sec:spanner}.\fi

\ifFull \else
\paragraph*{In the full version}
Since the novel technical contribution is in illustrating Properties P1 and P2 above, we focus on the arguments for these in this extended abstract.  Unfortunately, even providing the full proof of these properties would exceed the length of the extended abstract, but we aim to give an overview of the main ideas of the proof here. Full details of the proofs and omitted proofs can be found in the appended full version of the paper.
\fi

\section{Vertex-connectivity basics} \label{sec:vert-conn}

In this section, we consider minimal $(Q,r)$-vertex-connected graphs for a subset $Q$ of vertices and a requirement function $r:Q\to \{2,3\}$.
\ifFull
A subgraph induced by $S\subseteq V(G)$ or $S\subseteq E(G)$ in $G$ is denoted by $G[S]$.  The degree of vertex $v$ in $G$ is denoted by $d_G(v)$. 
By $P[a,b]$ or $T[a,b]$ we denote the $a$-to-$b$ subpath of path $P$ or tree $T$.

\begin{lemma}\label{lem: biconnected}
A minimal $(Q,r)$-vertex-connected graph is biconnected.
\end{lemma}
\begin{proof}
For a contradiction, assume that a minimal $(Q,r)$-vertex-connected graph $H$ has a cut-vertex $u$ and let the subgraphs be $H_i$ for $0<i\leq k$ and $k\geq2$  after removing $u$. Then for any vertex $x\in H_i$ and $y\in H_j$ ($i\neq j$), every $x$-to-$y$ path must contain $u$. If $H_i$ and $H_j$ ($i\neq j$) both have terminals, then terminals in those different subgraphs do not achieve the required vertex-connectivity. It follows that there exists one subgraph $H_i\cup\{u\}$ that contains all the terminals. For any two terminals $x,y\in H_i\cup\{u\}$, the paths witnessing their connectivity are simple and so can only visit $u$ once. Therefore, $H_i\cup\{u\}$ is a smaller subgraph that is $(Q,r)$-vertex-connected, contradicting the minimality.
\end{proof}

\subsection{Ear decompositions}\label{sec:eardecomp}  An {\em ear decomposition} of a graph is a partition of its edges into a sequences of cycles and paths (the {\em ears} of the decomposition) such that the endpoints of each ear belong to union of earlier ears in the decomposition. An ear is {\em open} if its two endpoints are distinct from each other. An ear decomposition is {\em open} if all ears but the first are open. A graph containing more than one vertex is biconnected if and only if it has an open ear decomposition~\cite{Whitney32}. Ear decompositions can be found greedily starting with any cycle as the first ear.  It is easy to see that a more general ear decomposition can start with any biconnected subgraph:

\begin{observation}\label{obs: ear}
  For any biconnected subgraph $H$ of a biconnected graph $G$, there exists an open ear decomposition $E_1, E_2, \dots, E_k$ of $G$ such that $H=\bigcup_{i\le j} E_i$ for some $j \leq k$.
\end{observation} 
Let $G$ be a minimal $(Q, r)$-vertex-connected graph, and let $H$ be a minimal $Q_3$-triconnected subgraph of $G$ where $Q_3= r^{-1}(3)$. Then
more strongly, we can assume that each ear of $G$ that are within the parts of $G\setminus H$ contains a terminal.  ($G$ is biconnected by Lemma~\ref{lem: biconnected}.) We do so by starting with an open ear decomposition of $H$ and then for each terminal that is not yet spanned in turn, we add an ear through it; such an ear exists because these terminals require biconnectivity and must have two disjoint paths to the partially constructed ear decomposition.  Any remaining edges after the terminals have been spanned would contradict the minimality of $G$.   Formally and more specifically:
\begin{observation}\label{obs: terminal ear}
  For $G$ and $H$, there is an open ear decomposition $E_1, E_2, \dots, E_k$ of $G$ such that for any component $\chi$ of $G\setminus H$, $\chi = \bigcup_{i = a}^bE_i$ for some $a \le b \le k$ and $E_i$ contains a terminal for $i = a, \ldots, b$.
\end{observation}

\begin{lemma}\label{lem: ear}
  For $G$ and $H$, there is an open ear decomposition $E_1, E_2, \dots, E_k$ of $G$ such that for any component $\chi=\bigcup_{i=a}^b E_i$ of $G\setminus H$, any path in $\chi$ (or $\chi\setminus E_a$) with both endpoints in $H$ (or $H\cup E_a$) strictly contains a vertex of $Q$. 
\end{lemma}
\begin{proof}
We prove the lemma for the path in $\chi$ with endpoints in $H$; the other case can be proved similarly.
We prove this by induction on the index of the ear decomposition guaranteed by Observation~\ref{obs: terminal ear}.  A path $P$ in $G\setminus H$ with both endpoints in $H$ belongs to a component $\chi$ of $G\setminus H$.  
Suppose that the lemma is true for every $H$-to-$H$ path in $\bigcup_{i = a}^cE_i$; we prove the lemma true for such a path in $\bigcup_{i = a}^{c+1}E_i$. Since $E_{c+1}$ is an open ear, any path with two endpoints in $H$ that uses an edge of $E_{c+1}$ would have to contain the entirety of $E_{c+1}$, which contains a terminal.
\end{proof}

\subsection{Removable edges}\label{sec: removable_edges}

\fi

\ifFull
\else
Borradaile and Klein prove that in a minimal $Q$-biconnected graph, every cycle contains a terminal (Theorem 2.5~\cite{BK13}).  We prove a similar property (in the appended full paper) for a minimal $(Q,r)$-vertex connected graph here. 
This property implies property P1, that is the intersection of an optimal solution with the interior of any brick can be partitioned into a set of trees whose leaves are on the boundary of the brick.
Note that our proof for this property does not depend on planarity.

For a $Q$-triconnected graph $H$, we can obtain another graph $H'$ by contracting all the edges incident to the vertices of degree two in $H$.  We say $H'$ is {\em contracted version of $H$} and, alternatively, is \emph{contracted} $Q$-triconnected.  
We can prove that $H'$ is triconnected.
\fi

Holton, Jackson, Saito and Wormald study the {\em removability} of edges in triconnected graphs~\cite{HJSW90}.
For an edge $uv$ of a simple, triconnected graph $G$, removing $e$ consists of:
\begin{itemize}
\item Deleting $uv$ from $G$.
\item If $u$ or $v$ now have degree 2, contract incident edges.
\item If there are now parallel edges, replace them with a single edge.
\end{itemize}
\ifFull
\else
By applying several results of Holton et al.~\cite{HJSW90} about removable edges, we prove (in the appended full paper) that every cycle in a minimum contracted $Q$-triconnected graph contains a terminal.
For a graph $G$ that is $(Q,r)$-vertex connected, let $G'$ be a minimum $Q$-triconnected graph that is a supergraph of $G$.
Let $G''$ is the contracted version of $G'$.
Then every cycle in $G''$ contains a terminal. 
Since $G'$ is a subdivision of $G''$, we know every cycle in $G'$ contains a terminal.
Since $G$ is a subgraph of $G'$, we have the following theorem.  
\fi 
\ifFull
The resulting graph is denoted by $G\ominus e$. If $G\ominus e$ is triconnected, then $e$ is said to be removable.  We use the following theorems of Holton et al.~\cite{HJSW90}.
\begin{theorem}[Theorem 1~\cite{HJSW90}]\label{thm: nonremovable}
 Let $G$ be a triconnected graph of order at least six and $e\in E(G)$. Then $e$ is nonremovable if and only if there exists a set $S$ containing exactly two vertices such that $G\setminus \{e, S\}$ has exactly two components $A,B$ with $|A|\geq2$ and $|B|\geq2$.
\end{theorem}
In the above theorem, we call $(e,S)$ a separating pair.  For a separating pair $(e,S)$ of $G$, we say $S$ is the {\em separating set} for $e$. 
\begin{theorem}[Theorem 2 \cite{HJSW90}] \label{thm: removableone}
Let $G$ be a triconnected graph of order at least six, and let $(e,S)$ be a separating pair of $G$. Let $e=xy$, and let $A$ and $B$ be the two components of $G\setminus \{e, S\}$, $x\in A$, and $y\in B$. Then every edge joining $S$ and $\{x,y\}$ is removable.
\end{theorem}
\begin{theorem}[Theorem 6 part (a) \cite{HJSW90}]\label{thm: nonremovable_cycle} 
Let $G$ be a triconnected graph of order at least six and $C$ be a cycle of $G$. Suppose that no edges of $C$ are removable. Then there is an edge $yz$ in $C$ and a vertex $x$ of $G$ such that $xy$ and $xz$ are removable edges of $G$, $d_G(y)=d_G(z)=3$ and $d_G(x)\geq4$.
\end{theorem}

\subsection{Properties of minimal $(Q,r)$-vertex-connected graphs}\label{sec: properties}

For a $Q$-triconnected graph $H$, we can obtain another graph $H'$ by contracting all the edges incident to the vertices of degree two in $H$.  We say $H'$ is {\em contracted version of $H$} and, alternatively, is \emph{contracted} $Q$-triconnected.  
\begin{lemma}\label{lem: triconnected}
A contracted minimal $Q$-triconnected graph is triconnected. 
\end{lemma}
\begin{proof}
For a contradiction, we assume that a contracted minimal $Q$-triconnected graph $H$ has a pair of cut-vertices $\{u,v\}$ and let the subgraphs be $H_i$ for $0<i\leq k$ and $k\geq2$ after removing $\{u,v\}$. Then for any vertices $x\in H_i$ and $y\in H_j$ ($i\neq j$), every $x$-to-$y$ path must use either $u$ or $v$. If $H_i$ and $H_j$ ($i\neq j$) both contain terminals, then terminals in those different subgraphs do not satisfy the triconnectivity. It follows there exists one strict subgraph $H_i\cup\{u,v\}$ containing all the terminals. For any two terminals $x,y\in H_i\cup\{u,v\}$, the paths witnessing their connnectivity are simple and can only visit $u$ and $v$ once. So $H_i\cup\{u,v\}$ is a smaller subgraph that is $Q$-triconnected, which contradicts the minimality.
\end{proof} 
\begin{lemma}\label{lem: minimal}
If $H$ is a minimal $Q$-triconnected graph, then the contracted version of $H$ is also a minimal $Q$-triconnected graph. 
\end{lemma}
\begin{proof}
Let $H'$ be the contracted $Q$-triconnected graph obtained from $H$. For a contradiction, assume $H'$ is not minimal $Q$-triconnected. Then we can delete at least one edge, say $e$, in $H'$ while maintaining $Q$-triconnectivity. Then $e$ corresponds a path in $H$, deleting which will not affect the $Q$-triconnectivity. This contradicts the minimality of $H$.  
\end{proof}

\begin{lemma}\label{lem: simple}
For $|Q|=3$, a contracted minimal $Q$-triconnected graph is simple or a triangle with three pairs of parallel edges. For $|Q|>3$, a contracted minimal $Q$-triconnected graph is simple.
\end{lemma}
\begin{proof}
Let $H$ be a minimal $Q$-triconnected graph, and $H'$ the contracted graph. We have the following observation.
\begin{observation}\label{obs: parallel}
If there are parallel edges between any pair of vertices in $H'$, the paths witnessing the vertex-connectivity between any other pair of terminals can only use one of the parallel edges. 
\end{observation}
By the above observation, the parallel edges can only be between terminals in $H'$. 
\begin{claim}\label{clm: noparallel}
For $|Q|>2$, there can not exist three parallel edges between any two terminals in $H'$.
\end{claim}
\begin{proof}
Let $t_1$, $t_2$ and $t_3$ be three terminals and assume there are three parallel edges, say $e_1$, $e_2$ and $e_3$, between $t_1$ and $t_2$. Let $P_1$ be the path in $H$ corresponding to $e_1$. We will argue that $H\setminus P_1$ is $Q$-triconnnected by showing that $H'-e_1$ is $Q$-triconnected. There must be a path in $H'$ from $t_1$ to $t_3$ that does not use $e_1$, $e_2$ and $e_3$ by Observation~\ref{obs: parallel} and a path in $H'$ from $t_2$ to $t_3$ that does not use $e_1$, $e_2$ and $e_3$. These paths witness a $t_1$ to $t_2$ path $R$ in $H'$ that does not use $e_1$, $e_2$ and $e_3$. So after deleting $e_1$, $t_1$ and $t_2$ are still triconnected by $e_2$, $e_3$ and $R$. By Observation~\ref{obs: parallel}, deleting $e_1$ does not affect the triconnectivity of other pairs of terminals.
\end{proof}
We first prove the first statement of Lemma~\ref{lem: simple}. Let $t_1$, $t_2$ and $t_3$ be the three terminals, and suppose $H'$ is not simple; let $e_1$ and $e_2$ be parallel edges w.l.o.g.~between $t_1$ and $t_2$. By Observation~\ref{obs: parallel}, there must be two disjoint paths in $H'$ from $t_1$ to $t_3$ that do not use $e_1$ and $e_2$. Therefore, there is a simple cycle, called $C_{13}$, through $t_1$ and $t_3$. Similarly, there is another cycle, called $C_{23}$, through $t_2$ and $t_3$. If $C_{13}$ and $C_{23}$ have only one common vertex $t_3$, then $C_{13}\cup C_{23}\cup\{e_1,e_2\}$ is a subgraph of $H'$ that is $Q$-triconnected, and must be a triangle with three pairs of parallel edges by the minimality of $H'$ and Lemma \ref{lem: minimal}. If $C_{13}$ and $C_{23}$ have more than one common vertex, then $C_{13}\cup C_{23}$ contains a simple cycle through $t_1$ and $t_2$. So $H'-e_1$ will be a smaller $Q$-triconnected graph, contradicting the minimality of $H'$.

Now we prove the second statement. For any four terminals $t_1,t_2,t_3,t_4\in Q$, without loss of generality, assume there are two parallel edges $e_1$ and $e_2$ between $t_1$ and $t_2$ in $H'$ by Claim~\ref{clm: noparallel}. We will prove that $H'-e_1$ is also $Q$-triconnected, contradicting the minimality of $H$. To prove this, we argue that $t_1$ and $t_2$ are four-vertex-connected in $H'$ and biconnected in $H'\setminus\{e_1,e_2\}$.

For a contradiction, we assume $t_1$ and $t_2$ are simply connected but not biconnected in $H'\setminus\{e_1,e_2\}$. By Claim~\ref{clm: noparallel}, every other pair of terminals is at least biconnected in $H'\setminus\{e_1,e_2\}$.
Consider the block-cut tree of $H'-e_1$, the tree whose 
vertices represent maximal biconnected components and whose edges represent shared vertices between those components~\cite{ET76}.

The biconnectivity of $t_1$ and $t_3$ implies $t_1$ and $t_3$ are in a common block $B_{13}$ in the block-cut tree. Likewise, $t_2$ and $t_3$ are in a common block $B_{23}$. If $t_2\in B_{13}$ then $t_1$ and $t_2$ are biconnected, a contradiction. Now we have $B_{13}\cap B_{23}=\{t_3\}$. The biconnectivity of $t_1$ and $t_4$ implies that $t_1$ and $t_4$ are in the same block. Since $t_3$ is a cut for $t_1$ and $t_2$, $t_3$ is also a cut for $t_4$ and $t_2$, which contradicts the biconnectivity of $t_2$ and $t_4$. So $t_1$ and $t_2$ are biconnected in $H'\setminus\{e_1,e_2\}$.
\end{proof}

\begin{lemma}\label{lem: deletable}
Let $H$ be a contracted minimal $Q$-triconnected simple graph. Then for any $e\in E(H)$, if neither of the two endpoints of $e$ are terminals, $e$ is nonremovable.
\end{lemma}
\begin{proof}
By Lemma \ref{lem: triconnected}, $H$ is triconnected. Let $e\in H$ be an edge, neither of whose endpoints are terminals. For a contradiction, suppose $e$ is removable, then $H\ominus e$ is triconnected. Since neither of the two endpoints of $e$ are terminals, $Q\subseteq V(H\ominus e)$. So $H\ominus e$ is a smaller $Q$-triconnected graph than $H$, contradicting minimality of $H$. 
\end{proof}
\subsection{Cycles must contain terminals}\label{sec:cycle_terminals}

Borradaile and Klein proved that in a minimal $Q$-biconnected graph, every cycle contains a terminal (Theorem 2.5~\cite{BK13}).  We prove similar properties here.

\begin{lemma}\label{lem: cycle_triconnected}
Let $H$ be a contracted minimal $Q$-triconnected graph. Then every cycle in $H$ contains a vertex of $Q$.
\end{lemma}
\begin{proof}
$H$ is triconnected by Lemma \ref{lem: triconnected}.  
If $H$ is not simple, then either $|Q|=2$ or $|Q|=3$. If $|Q|=2$, then by the minimality of $H$, $H$ consists of three parallel edges. If $|Q|=3$, then by Lemma \ref{lem: simple}, $H$ is a triangle with three pairs of parallel edges.

Now we assume $H$ is a simple graph and by Lemma \ref{lem: simple}, $|Q|\geq3$. 
For a contradiction, assume there is a cycle $C$ in $H$ on which there is no terminal. Since $H$ is simple, by Lemma \ref{lem: deletable}, every edge on $C$ will be nonremovable. 
Since $|C|\geq 3$ and $|Q|\geq3$, $|V(H)|\geq6$. By Theorem \ref{thm: nonremovable_cycle} if there is no removable edges on $C$, then there exists an edge $yz$ on $C$ and another vertex $x$ of $H$ such that $xy$ and $xz$ are removable, $d_H(y)=d_H(z)=3$ and $d_H(x)\geq4$. 
In $H-xy$, $y$ is the only possible degree 2 vertex, and contracting $yz$ does not introduce any parallel edges, so $H\ominus xy$ contains $Q$ and is $Q$-triconnected, which contradicts the minimality of $H$.  
\end{proof} 
\fi
\begin{theorem}\label{thm: minimal} 
For a requirement function $r:Q\to \{2,3\}$, let $G$ be a minimal $(Q,r)$-vertex-connected graph. Then every cycle in $G$ contains a vertex of $Q$.  
\end{theorem}
\ifFull
\begin{proof}
Let $G'$ be a minimal $Q$-triconnected graph that is a supergraph of $G$. 
Let $G''$ be the contracted minimal $Q$-triconnected graph obtained from $G'$. Then by Lemma \ref{lem: cycle_triconnected}, every cycle in $G''$ contains a vertex of $Q$; $G'$ also has this property since $G'$ is a subdivision of $G''$ and clearly, $G$ as a subgraph of $G'$, also has this property. 
\end{proof}
\begin{lemma}\label{lem: oneterminal} 
Let $H$ be a contracted minimal $Q$-triconnected graph. If a cycle in $H$ only contains one vertex $v$ of $Q$, then one edge of that cycle incident to $v$ is removable.
\end{lemma}
\begin{proof}
For a contradiction, assume both edges in the cycle incident to $v$ are nonremovable. 
By Lemma~\ref{lem: deletable} all the other edges in the cycle are also nonremovable, so 
all the edges in the cycle are nonremovable.
By Theorem~\ref{thm: nonremovable_cycle}, there is an edge $yz$ in the cycle and a vertex $x$ in $H$ such that $xy$ and $xz$ are removable, $d_H(y) = d_H(z) = 3$ and $d_H(x) \ge 4$.
Then one of $y$ and $z$ can not be in $Q$.
W.l.o.g.~assume $y$ is not in $Q$. 
In $H-xy$, $y$ is the only possible degree 2 vertex, and contracting $yz$ does not introduce any parallel edges, so $H\ominus xy$ contains $Q$ and is $Q$-triconnected, which contradicts the minimality of $H$.  
\end{proof}
\fi

\section{Connectivity Separation}\label{sec:conn-sep}

In this section we continue to focus on vertex connectivity and prove the Connectivity Separation Theorem. 
The Connectivity Separation Theorem for biconnectivity follows easily from Theorem~\ref{thm: ntforbi}.  To see why, consider two paths $P_1$ and $P_2$ that witness the biconnectivity of two terminals $x$ and $y$.  For an edge of $P_1$ to be in the same terminal-bounded component as an edge of $P_2$, there would need to be a $P_1$-to-$P_2$ path that is terminal-free.  However, such a path must contain a terminal by Theorem~\ref{thm: ntforbi}.
Herein we mainly focus on triconnectivity.

For a requirement function $r:Q\to \{2,3\}$, let $G$ be a minimal $(Q,r)$-vertex-connected planar graph.
We say a subgraph is {\em terminal-free} if it is connected and does not contain any terminals.  It follows from Theorem~\ref{thm: minimal} that any terminal-free subgraph of $G$ is a tree.  We partition the edges of $G$ into {\em terminal-bounded} components as follows: a terminal-bounded component is either an edge connecting two terminals or is obtained from a {\em maximal} terminal-free tree $T$ by adding the edges from $T$ to its neighbors, all of which are terminals.  Theorem~\ref{thm: treecycle} will also show that any terminal-bounded subgraph is also a tree.

For a connected subgraph $\chi$ and an embedding with outer face containing no edge of $\chi$, let $C(\chi)$ be the simple cycle that strictly encloses the fewest faces and all edges of $\chi$, if such a cycle exists. (Note that $C(\chi)$ does not exist if there is no such choice for an outer face.)
In order to prove the Connectivity Separation Theorem for bi- and triconnectivity, we start with the following theorem:
\begin{theorem}[Tree Cycle Theorem]\label{thm: treecycle}
Let $T$ be a terminal-bounded component in a minimal $Q$-triconnected planar graph $H$. Then $T$ is a tree and $C(T)$ exists with the following properties
\begin{description}
\item [$(a)$] The internal vertices of $T$ are strictly inside of $C(T)$.
\item [$(b)$] All vertices strictly inside of $C(T)$ are on $T$.
\item [$(c)$] All leaves of $T$ are in $C(T)$.
\item [$(d)$] Any pair of distinct maximal terminal-free subpaths of $C(T)$ does not contain vertices of the same terminal-bounded tree.
\end{description}
\end{theorem}
\ifFull
Theorem~\ref{thm: twoterminals} follows from the Tree Cycle Theorem.
\begin{proof}[Proof of Theorem~\ref{thm: twoterminals}]
For a contradiction, assume there is a cycle in $H$ that only containing one terminal, then there is a terminal-bounded component containing that cycle, which can not be a tree, contradicting Tree Cycle Theorem.
\end{proof}
\else
We note that Theorem~\ref{thm: twoterminals} follows immediately from the Tree Cycle Theorem. \fi
We give an overview of the proof the Tree Cycle Theorem in Subsection~\ref{sec:treecycle}.  
First, let us see how the Tree Cycle Theorem implies the Connectivity Separation Theorem.

\subsection{The Tree Cycle Theorem implies the Connectivity Separation Theorem}

For a requirement function $r:Q\to\{2,3\}$, let $G$ be a minimal $(Q,r)$-vertex-connected planar graph. 
Let $Q_3$ be the set of terminals requiring triconnectivity, and let $H$ be a minimal $Q_3$-triconnected subgraph of $G$. 
Let $Q_2 = Q \setminus Q_3$.    
\ifFull
We will prove the theorem for different types of pairs of terminals, based on their connectivity requirement.
Lemma~\ref{lem: ear} allows us to focus on $H$ when considering terminals in $H$ (note that terminals of $Q_2$ may be in $H$), for a simple corollary of Lemma~\ref{lem: ear} is that if two trees are terminal-bounded in $H$, then they cannot be subtrees of the same terminal-bounded tree in $G$.  Note that if we consider a subset of the terminals in defining free-ness of terminals, the same properties will hold for $Q$, as adding terminals only further partitions terminal-bounded trees.
\else
Consider two terminals $x$ and $y$.  We sketch the proof here; full details are in the appended full version.

Suppose $x$ and $y$ only require biconnectivity.
For this case, we can find a simple cycle $C$ containing $x$ and $y$ such that every $C$-to-$C$ path contains a terminal as an internal vertex. As argued at the start of Section~\ref{sec:conn-sep}, this proves the Connectivity Separation Theorem for $x$ and $y$.

Suppose $x$ and $y$ require triconnectivity, that is $x,y \in Q_3$. 
Since graph $H$ is $Q_3$-triconnected, there are three internally vertex-disjoint paths from $x$ to $y$ in $H$.
We  modify these three paths such that they do not contain edges of the same terminal-bounded tree.
Suppose  all three paths contain some edges of a common terminal-bounded tree $T$.
By the Tree Cycle Theorem, there is a cycle $C(T)$ that contains all leaves of $T$ and all other vertices of $T$ are enclosed by $C(T)$. 
So all the tree paths must intersect cycle $C(T)$.  Note that since both of $x$ and $y$ are terminals, the edges incident to $x$ and $y$ are not in the same terminal-bounded tree.
So, for each $x$-to-$y$ path, we can identify non-trivial subpaths: one to-$C(T)$ prefix and one from-$C(T)$ suffix. 
We can find two subpaths of $C(T)$ and one path in $T$ such that they are vertex-disjoint and the union of these three subpaths together with the to-$C(T)$ prefices and the from-$C(T)$ suffices defines another three internally vertex-disjoint $x$-to-$y$ paths in $H$.
Only one of the three new paths will contain edges of $T$.
By property (d) of the Tree Cycle Theorem, the two subpaths of $C(T)$ will not introduce any {\em shared} terminal-bounded tree.
We can apply a similar modification when there are only two $x$-to-$y$ paths containing edges of the same terminal-bounded tree.
The argument for extending the property from $H$ to $G$ requires minimal extra work.
\fi

\ifFull
\subsubsection{Connectivity Separation for $\mathbf{x,y \in Q_3}$} \label{sec:conn-sep-tric} 

For now, we consider only $Q_3$ to be terminals. 
We say connected components {\em share} a terminal-bounded tree if they contain edges (and so internal vertices) of that tree.  
We prove the following lemma which can be seen as a generalization of Connectivity Separation for contracted triconnected graphs.  We  use this generalization to prove Connectivity Separation for terminals both of which may not be in $Q_3$.  Connectivity Separation for terminals in $Q_3$ follows from this lemma by considering  three vertex disjoint paths paths matching $A$ to $B$ where $A = \{x,x,x\}$ and $B = \{y,y,y\}$.   Note that the lemma may swap endpoints of paths; in particular, for vertex disjoint $a$-to-$b$ and $c$-to-$d$ paths, the lemma may only guarantee $a$-to-$d$ and $c$-to-$b$ paths that do not share terminal-bounded trees.

\begin{lemma}\label{lem: disjointpaths}
If two multisets of vertices $A$ and $B$ (where $|A| = |B| = 2$, resp. $3$) satisfy the following conditions, 
then there are two (resp. three) internally vertex-disjoint paths from $A$ to $B$ such that no two of them share the same terminal-bounded tree.
\begin{enumerate}
\item Distinct vertices in $A$ are in distinct terminal-bounded trees and distinct vertices in $B$ are in distinct terminal-bounded trees.
\item There are two (resp. three) vertex-disjoint paths from $A$ to $B$.
\end{enumerate}
\end{lemma}
\begin{proof}
We prove the lemma for three paths as the two-paths version is proved by the first case of three-paths version.
Let $P_1$, $P_2$ and $P_3$ be the three vertex disjoint paths whose endpoints are in different terminal-bounded trees.  For $i=1,2,3$, let $a_i$ and $b_i$ be the endpoints of $P_i$.
Let $\mathcal{T}$ be the collection of terminal-bounded trees shared by two or more of $P_1, P_2$ and $P_3$.  We prove, by induction on the size of $\mathcal{T}$, that we can modify the paths to satisfy Connectivity Separation.  We pick a tree $T \in \mathcal{T}$ shared by (w.l.o.g.) $P_1$ and 
modify the paths so that $T$ is not shared by the new paths; further, we show that the terminal-bounded trees shared by the new paths are a subset of $\mathcal{T}\setminus \{T\}$.

We order the vertices of $P_i$ from $a_i$ to $b_i$ for $i=1,2,3$.  Among all the trees in $\mathcal{T}$ shared by $P_1$, let $T$ be the tree sharing the first vertex from $a_1$ to $b_1$ along $P_1$.  
Without loss of generality, assume that $T$ is shared by $P_2$.  ($P_3$ may also share $T$.)  Let $C$ be the cycle guaranteed by the Tree Cycle Theorem for $T$ in $H$.  Since $x$ and $y$ are terminals, they are not strictly enclosed by $C$ (property~(b)). Further, $P_1$ and $P_2$ must both contain vertices of $C$ because $P_1$ and $P_2$ both contain internal vertices of $T$ and the internal vertices of $T$ are strictly enclosed by $C$ (property~(a)).

\noindent{\bf Augmenting the paths} 
First we augment each path to simplify the construction, but may make the paths non-simple.  If there is a terminal-bounded tree shared by only $P_i$ and $C$ (and not $P_j$, $j \ne i$), then there is a terminal-free path from $P_i$ to $C$; add two copies of this path to $P_i$ ($P_i$ travels back and forth along this path).  We repeat this for every possible shared tree and $i = 1,2,3$.  We let $P_1,P_2,P_3$ be the resulting paths.  Note that adding such paths does not introduce any new shared terminal-bounded trees between $P_1,P_2,P_3$ and $P_1,P_2,P_3$ are still vertex-disjoint.  

Let $u_i$ and $v_i$ be the first and last vertex of $P_i$ that is in $C$.  There are two cases:

\noindent{\bf Case 1: $\mathbf{P_3}$ is disjoint from $\mathbf C$.} In this case there were no applicable augmentations to $P_3$ as described above (and $u_3$ and $v_3$ are undefined).  Since $T$ is the first tree in $\mathcal T$ along $P_1$, $u_1$ and $u_2$ cannot be internal vertices of the same terminal-bounded tree.  Further, by planarity and disjointness of the three paths, $\{u_1, v_1\}$ and $\{u_2, v_2\}$ do not interleave around $C$. Let $C_i$ be the $u_i$-to-$v_i$ subpath of $C$ disjoint from $\{u_{3-i}, v_{3-i}\}$ for $i=1,2$: $C_1$ and $C_2$ are disjoint.  By construction, $P_1 \cup P_2 \cup C_1 \cup C_2 \setminus \{P_1[u_1,v_1],P_2[u_2,v_2]\}$ contains two disjoint $A$-to-$B$ paths; these paths replace $P_1$ and $P_2$.  By the definition of $u_i$ and $v_i$ for $i=1,2$ and the above-described path augmentation, if $C_i$ shares a terminal-bounded tree with another path, then that tree was already in $\mathcal T$.  Therefore we have reduced the number of shared terminal-bounded trees (since $T$ is no longer shared) without introducing any new shared terminal-bounded trees.

\noindent{\bf Case 2: $\mathbf{P_3}$ and $\mathbf C$ have at least one common vertex.} In this case, $P_3$ may or may not contain internal vertices of $T$.  By planarity and disjointness of the paths, the sets $\{u_1,u_2,u_3\}$ and $\{v_1,v_2,v_3\}$ do not interleave around $C$.  Let $C_u$ and $C_v$ be the minimal subpaths of $C$ that span $\{u_1, u_2, u_3\}$ and $\{v_1, v_2, v_3\}$, respectively. Let $C_1$ and $C_2$ be the components of $C\setminus \{C_u, C_v\}$. Then $C_1$ and $C_2$ are disjoint paths that connect two vertices of $\{u_1, u_2, u_3\}$ to two vertices of $\{v_1, v_2, v_3\}$.  

By planarity and disjointness of the paths, there must be a leaf of $T$ that is an internal vertex of $C_u$ in order for paths to reach $T$ from that leaf of $T$; the same property holds for $C_v$. Let $C_3$ be the simple path from the middle vertex of $\{u_1,u_2,u_3\}$ to the middle vertex of $\{v_1, v_2, v_3\}$ in $C_u \cup C_v \cup T$.  Then the to-$C$ prefices of $P_i$, $C_i$ and the from-$C$ suffices of $P_i$ (for $i=1,2,3$) together form three vertex-disjoint paths between the same endpoints.  

The resulting path that contains $C_3$ is the only of the three resulting paths that contains internal vertices of $T$ since $C_1$ and $C_2$ do not share the same terminal-bounded tree because the endpoints of $C_3\cap T$ are terminals (property (d) of the Tree Cycle Theorem).  

By the definition of $u_i$ and $v_i$ and the above-described path augmentation, if $C_i$ shares a terminal-bounded tree with another path, then that tree was already in $\mathcal T$.  Therefore we have reduced the number of shared terminal-bounded trees (since $T$ is no longer shared) without introducing any new shared terminal-bounded trees.
\end{proof}

\subsubsection{Connectivity Separation for $\mathbf{x,y \in H \cap Q}$} \label{sec:conn-sep-H}

Note that $H$ may span vertices of $Q_2$.  Let $P_1$ and $P_2$ be vertex-disjoint $x$-to-$y$ paths.  The first and last edges of $P_1$ and $P_2$ are in different terminal-bounded trees, since $x$ and $y$ are terminals.  If either $P_1$ or $P_2$ is an edge, then we could obtain Connectivity Separation for $x$ and $y$.  Otherwise, let $x_1$ and $x_2$ be $x$'s neighbors on $P_1$ and $P_2$ respectively; similarly define $y_1$ and $y_2$.  Let $P_1'$ and $P_2'$ be the paths guaranteed by Lemma~\ref{lem: disjointpaths} when applied to $P_1[x_1,y_1]$ and $P_2[x_2,y_2]$.  Then $P_1' \cup P_2' \cup \{xx_1,xx_2,yy_1,yy_2\}$ contain vertex disjoint $x$-to-$y$ paths that satisfy the requirements of Connectivity Separation Theorem.

\subsubsection{Connectivity Separation for $\mathbf{x \in Q_2 \setminus H, y \in Q}$} 
Since $x$ and $y$ only require biconnectivity, we will prove that there exists a simple cycle $C$ containing $x$ and $y$, such that every $C$-to-$C$ path strictly contains a terminal, from which Connectivity Separation follows as argued in the beginning of Section~\ref{sec:conn-sep}.  The following claim gives a sufficient condition for such cycle $C$.

\begin{claim}\label{clm: 3}
If every terminal-bounded tree of $H$ and every component of $G\setminus E(H)$ contain at most one strict subpath of $C$, then every $C$-to-$C$ path contains a terminal.
\end{claim}
\begin{proof}
Let $P$ be a $C$-to-$C$ path and let $C_{\chi}$ be the subpath of $C$ in the component $\chi$ of $G\setminus H$.  
Then $C_{\chi}$ has two endpoints in $H$ and it must contain a terminal by Lemma~\ref{lem: ear}.  
Every subpath of $P$ in $\chi$ has endpoints in $H\cup C_{\chi}$. We can take $C_{\chi}$ as the first ear for $\chi$, and then every subpath of $P$ in $\chi$ contains a terminal by Lemma~\ref{lem: ear}. 
So if $P$ contains an edge of $\chi$, we have the claim.

If $P$ does not contain an edge in $G\setminus H$, then $P\subseteq H$. 
Further, if $P$ does not contain any terminal, then it must be in some terminal-free tree $T$ by condition of the claim $T$ contains only one subpath $C_T$ of $C$. 
Since $P$ is a $C$-to-$C$ path, $C_T$ and $P$ form a cycle in $T$, which contradicts $T$ is a tree. So $P$ must contain a terminal.
\end{proof}

The following claim will allow us to use Lemma~\ref{lem: disjointpaths} on parts of the $x$-to-$y$ paths.

\begin{claim}\label{clm:2}
  Each connected component of $G\setminus H$ has at most one non-terminal vertex in common with any terminal-bounded tree of $H$.
\end{claim}

\begin{proof}
For a contradiction, suppose $a$ and $b$ are two non-terminal vertices of component $\chi$ of $G\setminus H$ that are both in some terminal-bounded tree $T$ of $H$.  Let $P$ be an $a$-to-$b$ path in $\chi$ and let $R$ be a maximal suffix of $T[a,b]$ every internal vertex of which has degree 2 in $G$.
We show that deleting $R$ maintains the required connectivity of $G$ and this contradicts minimality of $G$. 

Let $H'= (H\setminus R) \cup P$.  First we show that $H'$ is $Q_3$-triconnected. Since $a$ and $b$ are not terminals, they are not leaves of $T$.
Then $(T\setminus R)\cup P$ is a tree that contains all the terminals of $T$ as leaves.  Since $H$ satisfies Connectivity Separation, as argued (Section~\ref{sec:conn-sep-tric} and~\ref{sec:conn-sep-H}), any two terminals have three vertex disjoint paths, no two of which contain edges from the same terminal-bounded tree.  Therefore replacing $T$ with $(T\setminus R)\cup P$ will preserve the connectivity.

To prove $G\setminus R$ is biconnected, we construct an open ear decomposition. Let $R'$ be the maximal superpath of $R$ in $H$ every internal vertex of which has degree 2 in $H$. (Note that the endpoints of $R$ need not have degree $>2$ in $H$.) Since the contracted version of $H$ is triconnected, the endpoints of $R'$ are triconnected in $H$ and $R'$ becomes an edge in the contracted version of $H$. Therefore $H\setminus R'$ is biconnected.  Consider an ear decomposition of $G$ that starts with an ear decomposition of $H \setminus R'$ as guaranteed by Observation~\ref{obs: ear}.   
Since every internal vertex of $R$ has degree 2 in $G$ and since the endpoints of $R$ have degree $>2$ in $G$, we can greedily select ears of $(G \setminus H)\cup (R' \setminus R)$ to span the resting vertices.
\end{proof}

Let $\chi$ be the component of $G\setminus H$ that contains $x$, and let $a$ and $b$ be two vertices of $\chi \cap H$. 
Then $a$ and $b$ are in distinct terminal-bounded trees of $H$: if either of $a$ and $b$ is terminal, then it could be in at least two terminal-bounded trees; otherwise by Claim~\ref{clm:2} they can not be in the same terminal-bounded tree.
We consider an ear decomposition of $G$ that is guaranteed by Lemma~\ref{lem: ear}, starting with an ear decomposition of $H$ with consecutive ears composing $\chi$.
There are two cases.

\noindent{\bf Case 1: $\mathbf{y\in H}$} In this case we construct a simple cycle $C$ to satisfy Claim~\ref{clm: 3} containing $x$ and $y$ as follows. 
We first find an $a$-to-$b$ path $P_1$ of $H$ that contains $y$: add another new vertex $t$ and two new edges $ta$ and $tb$ into $H$, then $H \cup \{ta, tb\}$ is biconnected and there exist two vertex disjoint paths from $t$ to $y$ through $a$ and $b$ respectively. 
So there exist two disjoint paths $P_a$ and $P_b$ from $A=\{y, y\}$ and $B = \{a, b\}$, and let $P_1 = P_a \cup P_b$.
Let $P_2$ be the $a$-to-$b$ path that is taken as the first ear of $\chi$.
If $P_a$ and $P_b$ share any terminal-bounded trees, we could modify the two paths by Lemma~\ref{lem: disjointpaths} for $A=\{y, y\}$ and $B=\{a, b\}$ such that the new paths $P_a'$ and $P_b'$ from $A$ to $B$ are vertex disjoint and do not share any terminal-bounded tree. 
Further, we shortcut $P_a'$ and $P_b'$ such that they have at most one subpath in each terminal-bounded tree.  
Since $y$ is a terminal, $C= P_a' \cup P_b' \cup P_2$ satisfies the conditions of Claim~\ref{clm: 3}, giving the Connectivity Separation.

\noindent{\bf Case 2: $\mathbf{y\in G\setminus H}$} In this case $x$ and $y$ may or may not be in the same component of $G\setminus H$. 

Suppose $x$ and $y$ are in the same component $\chi$ of $G\setminus H$. 

If there is an $a$-to-$b$ path $P_1$ in $\chi$ that contains both of $x$ and $y$, we take $P_1$ as the first ear of $\chi$.
Let $P_2$ be an $a$-to-$b$ path in $H$. We shortcut $P_2$ in each terminal-bounded tree such that each terminal-bounded tree contains at most one subpath of $P_2$. Let the cycle $C$ be composed by $P_1$ and $P_2$. Then by Claim~\ref{clm: 3}, every $C$-to-$C$ path contains a terminal, giving the Connectivity Separation.

If there is no such $P_1$, we take as the first ear of $\chi$ an $a$-to-$b$ path that contains $x$, and take the second ear containing $y$. Then there is a cycle $C$ containing $x$ and $y$ in the first two ears. 
Let $P$ be any $C$-to-$C$ path. 
Since any pair of vertices in $\chi \cap H$ are not in the same terminal-bounded tree (as argued for $a$ and $b$), $P$ will contain a terminal if it contains an edge in $H$. 
If $P$ does not contain an edge of $H$, then $P$ will contain a terminal by a similar proof of Lemma~\ref{lem: ear}.
So $P$ always contains a terminal, giving the Connectivity Separation. 

Suppose $x$ and $y$ are not in the same component. 
Let $\chi'$ be the component of $G\setminus H$ that contains $y$ and let $a'$ and $b'$ be two common vertices of $\chi'$ and $H$. Then $a'$ and $b'$ are not in the same terminal-bounded tree of $H$ as previously argued.
We argue there exist two disjoint paths from $A = \{a, b\}$ to $B = \{a', b'\}$ in $H$: add two new vertices $s$ and $t$, and four new edges $sa$, $sb$, $ta'$ and $tb'$ into $H$, then $H \cup \{sa, sb, ta', tb'\}$ is biconnected and there are two vertex disjoint paths from $s$ to $t$ that contains two paths $P_a$ and $P_b$ from $A$ to $B$. 
Let $P_1$ (or $P_2$) be the $a$-to-$b$ (or $a'$-to-$b'$) path that is taken as the first ear of $\chi$ (or $\chi'$).
If $P_a$ and $P_b$ share any terminal-bounded trees, we can modify them by Lemma~\ref{lem: disjointpaths} for $A$ and $B$ such that the new paths $P_a'$ and $P_b'$ from $A$ to $B$ are vertex disjoint and share no terminal-bounded tree. 
Further, we can shortcut $P_a'$ and $P_b'$ such that each terminal-bounded tree contains at most one subpath of $P_a'$ and $P_b'$. 
Let the cycle $C$ be composed by $P_a'$, $P_b'$, $P_1$ and $P_2$. Then every terminal-bounded tree of $H$ and every component of $G\setminus H$ contain at most one strict subpath of $C$. By Claim~\ref{clm: 3} every $C$-to-$C$ path contains a terminal, giving the Connectivity Separation.

This completes the proof of the Connectivity Separation Theorem. 
\fi

\subsection{Proof of Tree Cycle Theorem} \label{sec:treecycle}

Let $G$ be a minimal $Q_3$-triconnected planar graph.  We prove the Tree Cycle Theorem for the contracted $Q_3$-triconnected graph $H$ obtained from $G$. If the theorem is true for $H$, then it is true for $G$ since subdivision will maintain the properties of the theorem. 
\ifFull
By Lemmas~\ref{lem: minimal} and~\ref{lem: simple}, if there are parallel edges in $H$, then either $|Q_3| = 2$ and $H$ consists of three parallel edges or $|Q_3| = 3$ and $H$ is a triangle with three pairs of parallel edges. The Tree Cycle Theorem is trivial for these two cases.  
\paragraph{Proof Overview}
\else
We give a high-level overview of the proof; the full details of the proof are in the appended full version of the paper.

\fi
We focus on a maximal terminal-free tree $T^*$, rooted arbitrarily, of $H$ and the corresponding terminal-bounded component $T$ (that is, $T^* \subset T$). 
We view $T^*$ as a set $\cal P$ of root-to-leaf paths. We show that we find a cycle for each path in $\cal P$ that strictly encloses only vertices on the paths. 
The outer cycle of the cycles for all the paths in $\cal P$ defines $C(T)$. See Figure~\ref{fig: treecycle}.
Property (a) directly follows from the construction. 
Property (b) is proved by induction on the number of root-to-leaf paths of $T$: when we add a new cycle for a path from $\cal P$, the new outer cycle will only strictly enclose vertices of the root-to-leaf paths so far considered.
After that, we show any two terminals are triconnected when $T$ is a tree: by modifying the three paths between terminals in a similar way to \ifFull Lemma~\ref{lem: disjointpaths}\else the proof for Connectivity Separation\fi, only one path will require edges in $T$.
Since $T$ is connected, this proves $T$ is a tree by minimality of $H$.
Combining the above properties and triconnectivity of $H$, we can obtain property (c).
Property (d) is proved by contradiction: if there is another terminal-bounded tree $T'$ that shares two terminal-free paths of $C(T)$, then there is a terminal-free path in $T'$. We can show there is a removable edge in this path of $T'$, contradicting \ifFull Lemma~\ref{lem: deletable}. \else the minimality of $H$.

\begin{figure}[ht]
\centering
\includegraphics[scale=1]{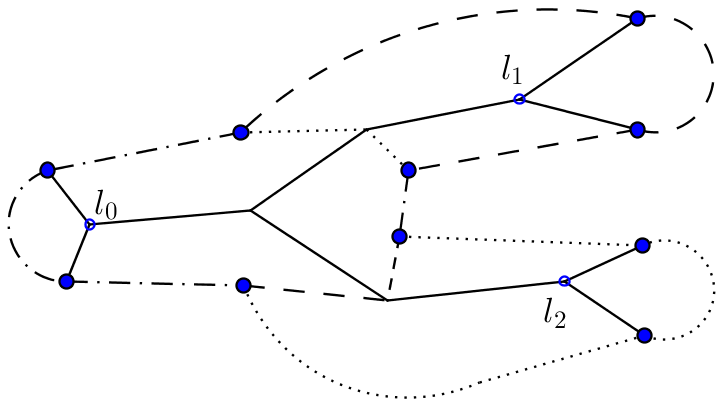}
\caption{Illustration of $C(T)$. The dashed cycle is $C_P$ for $P$ from $l_0$ to $l_1$ and the dotted cycle is $C_{P'}$ for $P'$ from $l_0$ to $l_2$. The outer boundary forms $C(T)$.}\label{fig: treecycle}
\end{figure}
\fi

\ifFull
\paragraph*{Property (a)}   
To prove that $C(T)$ exists, we will prove the following after proving some lemmas regarding terminal-free paths (Section~\ref{sec:strictly-term-free}), which guarantees that there is a drawing of $H$ such that $T^*$ is enclosed by some cycle:
\begin{lemma}\label{lem: exist}
There is a face of $H$ that does not touch any internal vertex of $T^*$.
\end{lemma}
We will also prove the following two lemmas in Section~\ref{sec:strictly-term-free}:
\begin{lemma}\label{lem: cycle}
For any terminal-free path $P$, there is a drawing of $H$ in which there is a simple cycle that strictly encloses all the vertices of $P$ and only the vertices of $P$.
\end{lemma}
\begin{lemma}\label{lem: disjoint}
Let $C_1$ and $C_2$ be two nested simple cycles of $H$ such that the edges of $C_1$ are enclosed by $C_2$, and $C_1$ and $C_2$ share at most one subpath.
Let $xy$ be an edge strictly enclosed by $C_2$ and not enclosed by (or on) $C_1$. 
If $H$ satisfies the following conditions, then $xy$ is removable:
\begin{description}
\item [1.] $C(xy)$ is vertex-disjoint with the common subpath of $C_1$ and $C_2$, and consists of two vertex-disjoint $C_1$-to-$C_2$ paths, a subpath of $C_1$ and a subpath of $C_2$ respectively. 
\item [2.] For every neighbor $z$ of $xy$ in $C(xy) \setminus \{C_1, C_2\}$, there is a $z$-to-$C_i$ path that shares only $z$ with $C(xy)$ (for $i=1$ or $i=2$).
\end{description}
\end{lemma}

Taking the face of $H$ guaranteed by Lemma~\ref{lem: exist} as the infinite face, for any path of $T^*$ this drawing guarantees a cycle as given by Lemma~\ref{lem: cycle}.  
Arbitrarily root $T^*$ and let $\cal P$ be a collection of root-to-leaf paths that minimally contains all the edges of $T^*$.  
For path $P \in {\cal P}$, 
let $C_P$ be the cycle that is guaranteed by Lemma~\ref{lem: cycle} which encloses the fewest faces.
By the maximality of $T^*$, the neighbors of $P$'s endpoints in $C_P$ are terminals.
Since $P$ is a path of $T^*$ and $T^*$ is a maximal terminal-free tree and since, by Lemma~\ref{lem: cycle}, $C_P$ strictly encloses only the vertices of $P$, the neighbors of $P$ on $C_P$ are either terminals or vertices of $T^*$.  

We construct $C(T)$ from $\bigcup_{P\in \cal P} C_P$.
Consider any order $P_1, P_2, \dots, P_{|\cal P|}$ of the paths in $\cal P$. 
Let $C^1 = C_{P_1}$ and 
let $C^i$ be the cycle bounding the outer face of 
$C^{i-1} \cup C_{P_{i}}$ for $i = 2, \dots, |P|$. 
Inductively, $C^{i-1}$ bounds a disk, and strictly encloses $P_1, \dots, P_{i-1}$. Also, $C_{P_i}$ bounds a disk that overlaps $C^{i-1}$'s disk. We define $C(T) = C^{|\cal P|}$. It follows that $C(T)$ strictly encloses $T^*$ (giving property (a)). That $C(T)$ encloses the fewest faces will follow from properties (b) and (c).
An example is given in Figure~\ref{fig: treecycle}.

\begin{figure}[ht]
\centering
\includegraphics[scale=1]{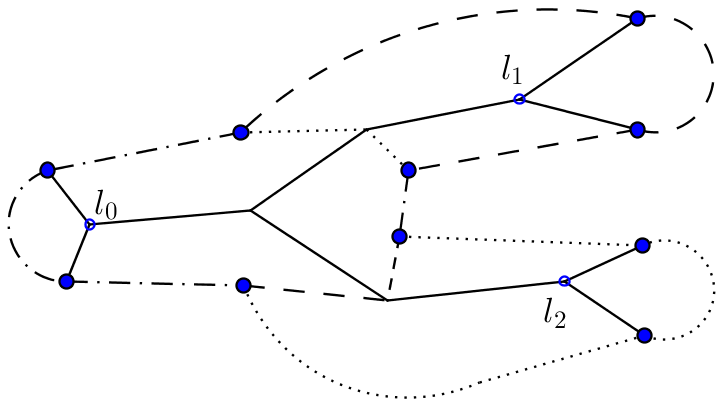}
\caption{Illustration of $C(T)$. The dashed cycle is $C_P$ for $P$ from $l_0$ to $l_1$ and the dotted cycle is $C_{P'}$ for $P'$ from $l_0$ to $l_2$. The outer boundary forms $C(T)$.}\label{fig: treecycle}
\end{figure}
\paragraph*{Property (b)}
We prove the following lemma and Property (b) follows from $C(T) = C^{|{\cal P}|}$.
\begin{lemma}\label{lem: cycles}
$C^i$ strictly encloses only the vertices of $\bigcup_{j\le i}P_j$ and encloses fewest faces.
\end{lemma}
\begin{proof} 
We prove this lemma by induction by
assuming the lemma is true for $C^{i-1}$.
The base case ($C^1$) follows from Lemma~\ref{lem: cycle}. 
Refer to Figure~\ref{fig: cycles} (a).
Since all the paths in $\cal P$ share the root $r$ of $T^*$ as one endpoint, $C^{i-1}$ and $C_{P_i}$ both strictly enclose $r$ and so must enclose a disk enclosing $r$.
This disk is bounded by a cycle $C$ consisting of four subpaths: two vertex-disjoint subpaths $R_1$ and $R_2$ of $C^{i-1}\cap C_{P_i}$ 
and subpaths $R_3 \subseteq C^{i-1}$ and $R_4 \subseteq C_{P_i}$ each connecting $R_1$ and $R_2$. 

\begin{figure}[ht]
\centering
\includegraphics[scale=1]{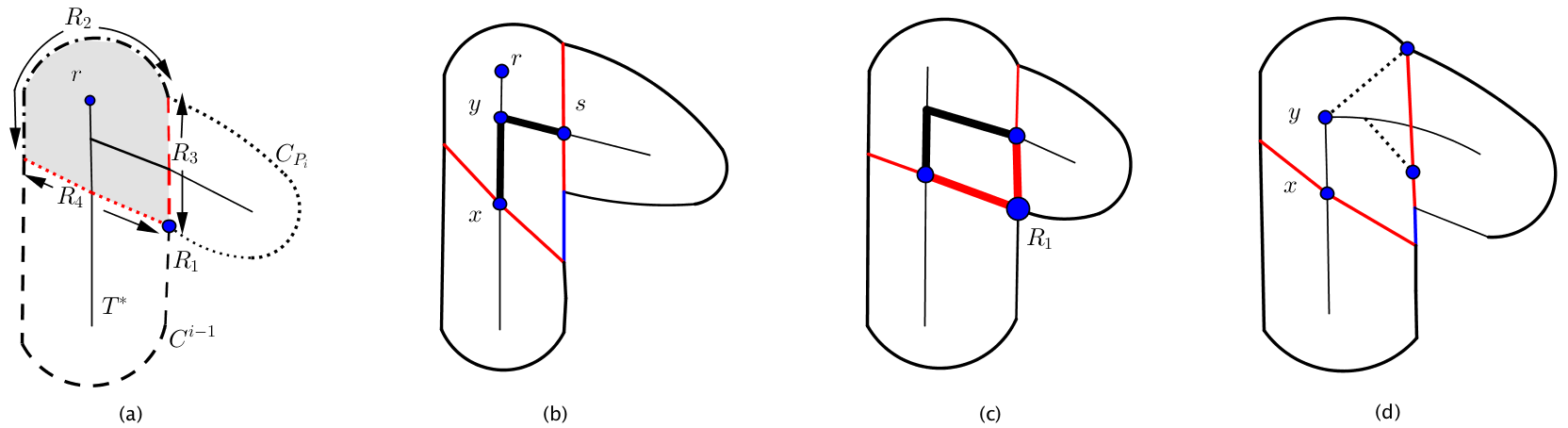}
\caption{Examples for Lemma~\ref{lem: cycles}. 
(a) The dotted cycle is $C_{P_i}$ and the dashed cycle is $C^{i-1}$. 
The two red paths are $R_3$ and $R_4$. $R_1$ is trivial.
The shaded region represents the common faces enclosed by $C^{i-1}$ and $C_{P_i}$. 
(b) The bold path is $P^*$ and $y$ is strictly inside of $C$.
(c) If $R_1$ is a vertex, then the bold cycle only contains one terminal: $R_1$, which has two neighbors in $\bigcup_{j\le i}P_j$.
(d) The dotted paths represent possible $y$-to-$R_3$ paths inside of $C$ that is different from $P_i$.
}\label{fig: cycles}
\end{figure}
Notice that $R_3$ must be enclosed by $C_{P_i}$ and contains a vertex of $T^*$, since $C^{i-1}$ strictly encloses only $\bigcup_{j<i} P_j$ and so is crossed by $P_i$ at some vertex in $R_3$. 
Similarly, $R_4$ must be enclosed by $C^{i-1}$ and contains a vertex in $T^*$.
So there is an $R_4$-to-$R_3$ path $P^*$ in $\bigcup_{j\le i} P_j$ enclosed by $C$.
Let $xy$ be the first edge in $P^*$ with $x\in R_4$
(Figure~\ref{fig: cycles} (b)).
Then we have the following observations about $xy$:
\begin{observation}\label{obs: xy}
Edge $xy$ is nonremovable (since neither $x$ nor $y$ are terminals, Lemma~\ref{lem: deletable}).
\end{observation}
\begin{observation}\label{obs: y}
Vertex $y$ is strictly enclosed by $C$.
\end{observation}
\begin{proof}
Let $s$ be the endpoint of $P^*$ in $R_3$. 
Then $s\in P_i$ since $s$ is an internal vertex of $R_3$ and strictly enclosed by $C_{P_i}$. 
The $r$-to-$s$ path $\Psi_1$ of $T^*$ is strictly enclosed by $C_{P_i}$, so $\Psi_1$ is a subpath of $P_i$.  
The $r$-to-$x$ path $\Psi_2$ of $T^*$ is strictly enclosed by $C^{i-1}$, so $\Psi_2$ is a subpath of $\bigcup_{j<i} P_j$.
Note the lowest common ancestor $LCA_{T^*}(x,s)$ of $x$ and $s$ in $T^*$ is in $\Psi_1\cap \Psi_2$.
So $LCA_{T^*}(x,s)\neq x$ since $LCA_{T^*}(x,s) \in \Psi_1\subseteq P_i$ and 
$LCA_{T^*}(x,s)\neq s$ since $LCA_{T^*}(x,s)\in \Psi_2\subseteq \bigcup_{j<i} P_j$.
See Figure~\ref{fig: cycles} (b). 
It follows $LCA_{T^*}(x,s)$ is strictly enclosed by $C$.
Since $y$ is between $x$ and $LCA_{T^*}(x,s)$ in $P^*$, we know $y\notin R_3$ and the claim follows.
\end{proof}

We argue that $R_1$ and $R_2$ must each contain at least one edge and so it will follow that $R_3$ and $R_4$ are vertex-disjoint.
For a contradiction, assume $R_1$ is a vertex. 
Refer to Figure~\ref{fig: cycles} (c).
Consider the cycle $C^*$ formed by $\bigcup_{j\le i} P_j$, a subpath of $R_3$, $R_1$ and a subpath of $R_4$. 
By Theorem~\ref{thm: minimal}, $C^*$ must contain a terminal. 
By construction, any terminal in $C^*$ must be in $R_1$. So $R_1$ is a terminal $t$.
Since $t$ is the crossing of $C^{i-1}$ and $C_{P_i}$, $t$'s degree is at least four. 
By Lemma~\ref{lem: oneterminal}, there is an edge in $C$ incident to $t$ that is removable, which contradicts the minimality of $H$. 
The same argument holds for $R_2$.
By the same reasoning, we have the following observation.
\begin{observation}\label{obs: neighbor}
Any vertex in $(C^{i-1}\cup C_{P_i}) \setminus \bigcup_{j\le i} P_j$ has at most one neighbor in $\bigcup_{j\le i} P_i$.
\end{observation}
Next, we will prove by contradiction that the subpath $S = C_{P_i} \setminus E( \bigcup_j R_j )$ only shares endpoints with $C^{i-1}$.
This will imply that $C^i$ strictly encloses only vertices of $P_i$ and vertices strictly inside of $C^{i-1}$.
Further, this will also imply that $C^i$ encloses fewest faces.
If there is another cycle $C^I$ that strictly encloses all vertices of $\bigcup_{j\le i} P_j$ and fewer faces than $C^i$, then there is some face that is enclosed by $C^i$ but not enclosed by $C^I$.
If that face is inside of $C^{i-1}$, then $C^{i-1}$ is not the cycle that encloses $\bigcup_{j<i} P_j$ and fewest faces, contradicting our inductive hypothesis; 
if that face is inside of $C_{P_i}$, then $C_{P_i}$ is not the cycle that encloses $P_i$ and fewest face, contradicting our choice of $C_{P_i}$.
So there can not be such cycle $C^I$, giving the lemma.

To prove that $S$ only shares endpoints with $C^{i-1}$, we first have some claims.

\begin{claim}\label{clm: R3}
The $y$-to-$R_3$ subpath of $P_i$
is the only one $y$-to-$R_3$ path whose edges are strictly enclosed by $C$.
\end{claim}
\begin{proof}
By Observation~\ref{obs: y}, $y$ is strictly enclosed by $C$. 
For a contradiction, assume there is another $y$-to-$R_3$ path $\Psi$ enclosed by $C$ that is not a subpath of $P_i$.
Notice that the endpoint of $\Psi$ in $R_3$ can not be an endpoint of $R_3$, for otherwise this endpoint of $R_3$ has two neighbors in $\bigcup_{j\le i} P_j$, contradicting Observation~\ref{obs: neighbor}.
Then we know $\Psi\subseteq T^*$, 
and we have two $y$-to-$R_3$ paths in $T^*$, which together with $R_3$ form a cycle without any terminals, contradicting Theorem~\ref{thm: minimal}.
See Figure~\ref{fig: cycles} (d).
\end{proof}

\begin{figure}[ht]
\centering
\includegraphics[scale=1]{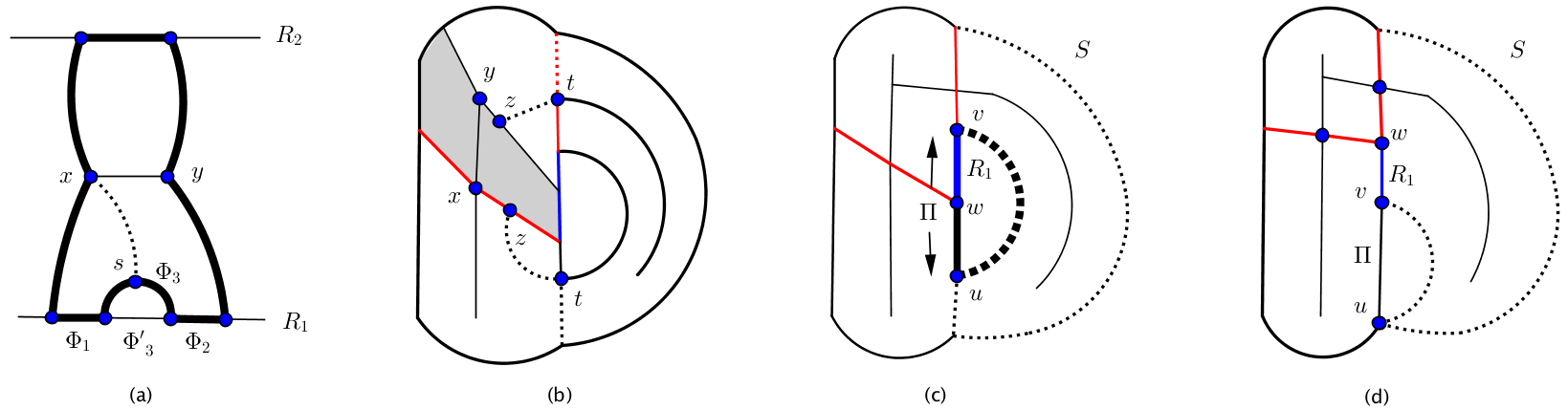}
\caption{Examples for Lemma~\ref{lem: cycles}. 
(a) The bold cycle is $C(xy)$.
(b) $C(xy)$ bounds the shaded region contains. The dotted paths are some possible examples for $\Phi$.
(c) The dotted path is $S$, which shares a subpath with $C^{i-1}$.
The bold cycle is $C'$. $R_3$ and $R_4$ are disjoint. $\Pi$ contains $R_1$.
(d) $\Pi$ does not contain $R_1$, and the endpoint $w$ of $R_1$ has two neighbors in $\bigcup_{j\le i} P_j$. 
}\label{fig: neighbors}
\end{figure}
\begin{claim}\label{clm: condition1}
$C(xy)$ consists of two vertex-disjoint $R_1$-to-$R_2$ paths, a subpath of $R_1$ and a subpath of $R_2$.
\end{claim}
\begin{proof}
We show the two $x$-to-$y$ subpaths of $C(xy)$ share at least an edge with $R_1$ and $R_2$ respectively, that implies $C(xy)$ must contain two disjoint $R_1$-to-$R_2$ paths: one through $x$ and the other through $y$. 

We first argue that the two $x$-to-$y$ paths in $C(xy)$ must each contain at least one terminal in $R_1\cup R_2$.
$C(xy)\cup xy$ contains two cycles whose intersection is $xy$. 
By Theorem~\ref{thm: minimal}, each of these cycles must contain a terminal.
Since neither $x$ nor $y$ are terminals, each of the $x$-to-$y$ paths in $C(xy)$ must contain a terminal.
Note that $C(xy)$ is enclosed by $C$ and 
there is no terminal strictly enclosed by $C$.
Further, since $R_3$ is enclosed by $C_{P_i}$, all its internal vertices are not terminals. Similarly, all internal vertices of $R_4$ are not terminals.
So all the terminals in $C$ are either in $R_1$ or $R_2$, and each $x$-to-$y$ path in $C(xy)$ contains at least one terminal in $R_1\cup R_2$.

Next, we show the two terminals can not be both in $R_1$ or $R_2$, 
which implies the two paths must share vertices with $R_1$ and $R_2$ respectively.
Notice that $P^*$ divides the region inside of $C$ into two parts, each of which has $R_1$ and $R_2$ in the bounding cycle respectively. 
If $C(xy)$ only contains vertices in $R_1$ or $R_2$, then $C(xy)$ must cross $P^*$ and there is another cycle that encloses $xy$ and fewer faces than $C(xy)$, contradicting the definition of $C(xy)$.
Then $C(xy)$ must contain two $R_1$-to-$R_2$ paths: one through $x$ and the other through $y$. 

Both of $C(xy)\cap R_1$ and $C(xy)\cap R_2$ can not be a vertex, for otherwise the vertex has two neighbors in $\bigcup_{j \le i}P_j$, contradicting Observation~\ref{obs: neighbor}. So the two $R_1$-to-$R_2$ paths are vertex disjoint.  

We then argue that $C(xy)$ could only share one subpath with $R_1$, and the same argument holds for $R_2$.
For a contradiction, assume there are more than one subpath of $C(xy)$ in $R_1$. 
Then let $\Phi_1$ and $\Phi_2$ be two such subpaths that are connected by a subpath $\Phi_3$ of $C(xy)$ that is not in $R_1$.
Refer to Figure~\ref{fig: neighbors} (a).
Notice that $\Phi_3$ is enclosed by a cycle $C'$ consisting of $R_1$, an $x$-to-$R_1$ subpath of $C(xy)$, an $y$-to-$R_1$ subpath of $C(xy)$ and edge $xy$.
Let $\Phi'_3$ be the subpath of $R_1$ that has the same endpoints as $\Phi_3$.
If $\Phi_3$ is an edge, then $(C^{i-1} \cup \Phi_3) \setminus \Phi'_3$ is a cycle that strictly encloses $\bigcup_{j<i} P_j$ and fewer faces than $C^{i-1}$, contradicting our inductive hypothesis.
If $\Phi_3$ contains an internal vertex $s$, then $s\in \bigcup_{j<i} P_j$ since it is strictly enclosed by $C^{i-1}$.
So there is an $s$-to-$x$ path in $\bigcup_{j<i} P_j$ whose edges must be enclosed by $C'$.
And then by replacing the $x$-to-$R_1$ subpath of $C(xy)$ with the $s$-to-$x$ path, we obtain a cycle that encloses $xy$ and fewer faces than $C(xy)$, contradicting the definition of $C(xy)$.
\end{proof}

\begin{claim}\label{clm: condition2}
For every neighbor $z$ of $xy$ in $C(xy)\setminus \{R_1, R_2\}$, there is a $z$-to-$(C^i\cup R_1)$ path that shares only $z$ with $C(xy)$.
\end{claim}
\begin{proof}
If any neighbor $z$ of $xy$ in $C(xy)$ is not in $R_1\cup R_2$, then 
$z$ must be on an $R_1$-to-$R_2$ subpath of $C(xy)$, whose internal vertices are in $\bigcup_{j\le i} P_j$. 
We construct a $z$-to-$(C^i\cup R_1)$ path $\Phi$ as follows: first find a $z$-to-$C^{i-1}$ subpath
$\Phi_1$ through $\bigcup_{j< i}P_j$, and then find a subpath $\Phi_2$ of $C^{i-1}$ that connects $\Phi_1$ with $C^i\cup R_1$ and is disjoint with $C(xy)$.
Let $t$ be the endpoint of $\Phi_1$ in $C^{i-1}$.
If $z\in C^{i-1}$, then $\Phi_1$ is empty and $t=z$.
If $t\in C^i\cup R_1$, then $\Phi_2$ is empty.
We define $\Phi = \Phi_1 \cup \Phi_2$.
In the following, we first argue that if $\Phi_1$ is not empty then it only shares $z$ with $C(xy)$, and then show that if $\Phi_2$ is not empty, then $\Phi$ only shares $z$ with $C(xy)$.

Assume $\Phi_1$ is not empty. Since $H$ is triconnected and $C(xy)$ encloses fewest faces, there exists a path $\Phi_1$ from $z$ to $C^{i-1}$ that is outside of $C(xy)$. 
Since $(\Phi_1 \setminus \{t\}) \subseteq \bigcup_{j<i} P_j$, $\Phi_1$ can not share any internal vertex with $C(xy)$.
Further, $\Phi_1$ and $C(xy)$ can not share $t$, for otherwise $t$ will be an endpoint of $C(xy)\cap R_1$ or $C(xy)\cap R_2$ and it has two neighbors in $\bigcup_{j\le i}P_j$: one via $\Phi_1$ and the other via $C(xy)$, contradicting Observation~\ref{obs: neighbor}.
So $\Phi_1$ only shares $z$ with $C(xy)$.

Consider the possible position of $t$ in $C^{i-1}$. If it is in $C^i\cup R_1$, then $\Phi_2$ is empty. 
If not, then it will be strictly enclosed by $C^i$. 
Refer to Figure~\ref{fig: neighbors} (b). 
There are two cases. 
\begin{description}
\item [Case 1.] If $t$ is in $R_3$, then we choose as $\Phi_2$ a subpath of $R_3$. 

Note that in this case $z$ could only be $y$'s neighbor by planarity and Claim~\ref{clm: condition1}.
Consider the position of $z$. 

If $z\notin R_3$, we argue 
$R_3$ and $C(xy)$ are vertex-disjoint, 
which will imply that there always exists a $t$-to-$C^i$ subpath of $R_3$.
For a contradiction, assume $R_3$ and $C(xy)$ are not disjoint.
Then $C(xy)$ must contains a $y$-to-$R_3$ subpath by planarity.
Further, $\{yz\} \cup \Phi_1$ witnesses another $y$-to-$R_3$ path whose edges are strictly enclosed by $C$, contradicting Claim~\ref{clm: R3}.
So $R_3$ and $C(xy)$ are vertex-disjoint.

If $z\in R_3$, we argue $R_3$ contains a $z$-to-$(C^i\cup R_1)$ subpath, which shares only $z$ with $C(xy)$.
By Claim~\ref{clm: condition1}, $C(xy)$ contains two $R_1$-to-$R_2$ paths, one of which contains $y$.
Then there are two $y$-to-$C$ subpaths of $C(xy)$: an $y$-to-$R_1$ subpath and an $y$-to-$R_2$ subpath, which are in distinct regions divided by $P^*$ and enclosed by $C$.
By Claim~\ref{clm: R3}, there is only one $y$-to-$R_3$ path whose edges are strictly enclosed by $C$. 
So only one of the $y$-to-$C$ subpath of $C(xy)$ shares vertices with $R_3$, and the other one is vertex-disjoint with $R_3$.
If the $y$-to-$R_i$ subpath is disjoint with $R_3$ for $i=1$ or $i=2$,
then there is a $z$-to-$R_i$ subpath of $R_3$ that shares only $z$ with $C(xy)$.
Since $R_2\subseteq C^i$, we have the $z$-to-$(C^i\cup R_1)$ subpath of $R_3$ that shares only $z$ with $C(xy)$.

\item [Case 2.] If $t$ is not in $R_3$, then we choose as $\Phi_2$ the $t$-to-$C^i$ subpath of $C^{i-1}$ that is enclosed in $C^i$. 
This subpath always exists and is disjoint with $C(xy)$. 
\end{description}
\end{proof}

Now we prove that $S$ only shares endpoints with $C^{i-1}$.
For a contradiction, assume $S$ has an internal vertex that is in $C^{i-1}$. 
Then by construction $C^i$ would enclose either $R_1$ or $R_2$;
w.l.o.g.~assume $R_1$ is enclosed by $C^i$.
Let $\Pi$ be the minimal subpath of $C^{i-1}$ that is enclosed by $C^i$ and connects an internal vertex $u$ of $S$ with the common endpoint $v$ of $R_1$ and $S$.
Let $w$ be the other endpoint of $R_1$.
There is a simple cycle $C'$ consisting of two $u$-to-$v$ subpaths: one is $\Pi$ and the other is a subpath of $S$.
See Figure~\ref{fig: neighbors} (c). 
Further, $\Pi$, and also $C'$, contains $R_1$, for otherwise $R_3$ and $R_4$ share $w$ as an endpoint and $w$ has two neighbors in $\bigcup_{j\le i}P_j$: one via $R_3$ and the other via $R_4$, contradicting Observation~\ref{obs: neighbor}. 
See Figure~\ref{fig: neighbors} (d).
It follows that $R_3$ and $R_4$ are vertex disjoint $C'$-to-$C^i$ paths.

We argue $H$ and $xy$ satisfy the conditions of Lemma~\ref{lem: disjoint} with $C_1 = C'$ and $C_2 = C^i$, which shows $xy$ is removable, giving a contradiction to Observation~\ref{obs: xy}:
\begin{description}
\item [Condition 1.]
Note that $C_1$ and $C_2$ could only share at most one vertex which is $u$ and shown in Figure~\ref{fig: neighbors} (d). 
Since $C_{P_i}$ is simple, $u\neq w$, so $u\notin C$. 
Since $C(xy)$ is enclosed by $C$, $u\notin C(xy)$. 
So $C(xy)$ does not contain the common vertex of $C_1$ and $C_2$. 
Then the first condition in Lemma~\ref{lem: disjoint} follows from
Claim~\ref{clm: condition1}.
\item [Condition 2.]
Since $R_1\subseteq C_1$, the second condition follows from Claim~\ref{clm: condition2}. 
Note that the $z$-to-$C^i$ path $\Phi$ constructed in Claim~\ref{clm: condition2} may contain a $z$-to-$C_1$ subpath, which shares only $z$ with $C(xy)$.
\end{description}
\end{proof}

\paragraph*{$T$ is a tree} To prove this, we show that when $T$ is a tree, $H$ is $Q_3$-triconnected. That is, for any pair of terminals $x$ and $y$, there are three $x$-to-$y$ internally vertex-disjoint paths only one of which contains internal vertices of $T$ when $T$ is a tree. Since $T$ is connected, this implies $T$ can not contain any cycle, for otherwise $H$ is not minimal. The proof is similar to that of Lemma~\ref{lem: disjointpaths} but simpler: we modify the paths between $x$ and $y$ such that only one of them contains internal vertices of $T$ while maintaining them internally vertex-disjoint. 

Let $R_1, R_2$ and $R_3$ be three $x$-to-$y$ disjoint paths. 
If there is only one path containing internal vertices of $T$, then 
it is sufficient for $T$ to be simply connected
for triconnectivity between $x$ and $y$. 
So we assume there are at least two paths containing internal vertices of $T$.
By Lemma~\ref{lem: cycles}, $C(T)$ strictly encloses all internal vertices of $T$, so any path that contains an internal vertex of $T$ must touch $C(T)$. 
We order the vertices of the three paths from $x$ to $y$.
Let $u_i$ and $v_i$ be the first and last vertex of $R_i$ that is in $C(T)$ for $i=1,2,3$. If $R_i$ is disjoint from $C(T)$, we say $u_i$ and $v_i$ is undefined. Let $C_1$ and $C_2$ be the disjoint minimum paths of $C(T)$ that connect $\{u_1, u_2, u_3\}$ and $\{v_1, v_2, v_3\}$.
If there are only two paths, say $R_1$ and $R_2$, containing vertices of $C(T)$, we can replace $R_1[u_1, v_1]\cup R_2[u_2, v_2]$ with $C_1 \cup C_2$. Then the new paths are disjoint and do not contain any internal vertex of $T$. If all the three paths contain vertices of $C(T)$, then 
let $C_3$ be the simple path from the middle vertex of $\{u_1, u_2, u_3\}$ to the middle vertex of $\{v_1, v_2, v_3\}$ in $(C(T)\cup T) \setminus (C_1\cup C_2)$.  Now the to-$C(T)$ prefices of $R_i$, $C_i$ and the from-$C(T)$ suffices of $R_i$ (for $i=1,2,3$) together form three vertex-disjoint paths between $x$ and $y$. Further, among the three resulting paths, only the path containing $C_3$ contains internal vertices of $T$. 
Therefore, it is sufficient for $T$ to be simply connected for $H$ to be $Q_3$-triconnected.

\paragraph*{Property~(c)}
By triconnectivity, every leaf of $T^*$ has at least two neighbors on $T$ that are terminals.
So each leaf of $T^*$ can not be a leaf of $T$ and then all leaves of $T$ are terminals.
By Lemma~\ref{lem: cycles}, 
$C(T)$ only strictly encloses all vertices of $T^*$.
So all its neighbors, which are terminals in $T$ by the maximality of $T^*$, are on $C(T)$, giving property (c).

Since each terminal-bounded component is a tree, any terminal on $T$ must be a leaf by the construction of terminal-bounded component.
Therefore, we have the following lemma.
\begin{lemma}\label{lem: leaf}
A vertex of a terminal-bounded tree $T$ is a terminal if and only if it is a leaf of $T$.
\end{lemma}

\paragraph*{Property~(d)}
Let $T_1$ be any terminal-bounded tree. 
For a contradiction, assume there is a terminal-bounded tree $T_2$ whose vertices are in distinct maximal terminal-free paths of $C(T_1)$. 
By Lemma~\ref{lem: leaf}, the leaves of $T_2$ are terminals.
So each component of $C(T_1) \cap T_2$ is a terminal-to-terminal path and contains at most one maximal terminal-free path. 
Then by the assumption for the contradiction, there are two vertex-disjoint non-trivial components (paths) $\Pi_1$ and $\Pi_2$ of $C(T_1) \cap T_2$.

Since $C(T_2)$ strictly encloses only internal vertices of $T_2$, the interiors of $C(T_2)$ and $C(T_1)$ overlap, and so $C(T_2)\cap T_1 \neq \emptyset$. 
Then $C(T_2)\cap T_1$ contains only paths whose endpoints are terminals on $C(T_1)$.
Let $\cal R$ be the set consisting of all maximal paths of $C(T_2)\cap T_1$.
Refer to Figure~\ref{fig: simplecycle} (a) and (b).
Since $C(T_2)$ is simple, and since the endpoints of paths in $\cal R$ are leaves of $T_1$, we have
\begin{observation}\label{obs: disjoint}
Any two paths in $\cal R$ are vertex disjoint. 
\end{observation}
Note that if $T_1$ is an edge or a star, $|{\cal R}| \le 1$ and so we would already have our contradiction. 

For any path $R_i \in \cal R$ with endpoints $u_i$ and $v_i$, there is an $u_i$-to-$v_i$ subpath $R'_i$ of $C(T_1)$ 
such that $R_i$ and $R'_i$ form a cycle enclosing a region that 
is enclosed by both of $C(T_1)$ and $C(T_2)$.
Since $C(T_1)$ and $C(T_2)$ only strictly enclose edges of $T_1$ and $T_2$ respectively, and since $T_1\neq T_2$, this region must be a face. 
Notice that any path of $C(T_1)\cap T_2$ could only be subpath of $R'_i$ for some $R_i\in \cal R$ for $C(T_2)$ encloses $T_2$.
By the following observation, there exist $R_1$ and $R_2$ such that $\Pi_1\subseteq R'_1$ and $\Pi_2 \subseteq R'_2$. 
\begin{figure}[ht]
\centering
\includegraphics[scale=0.9]{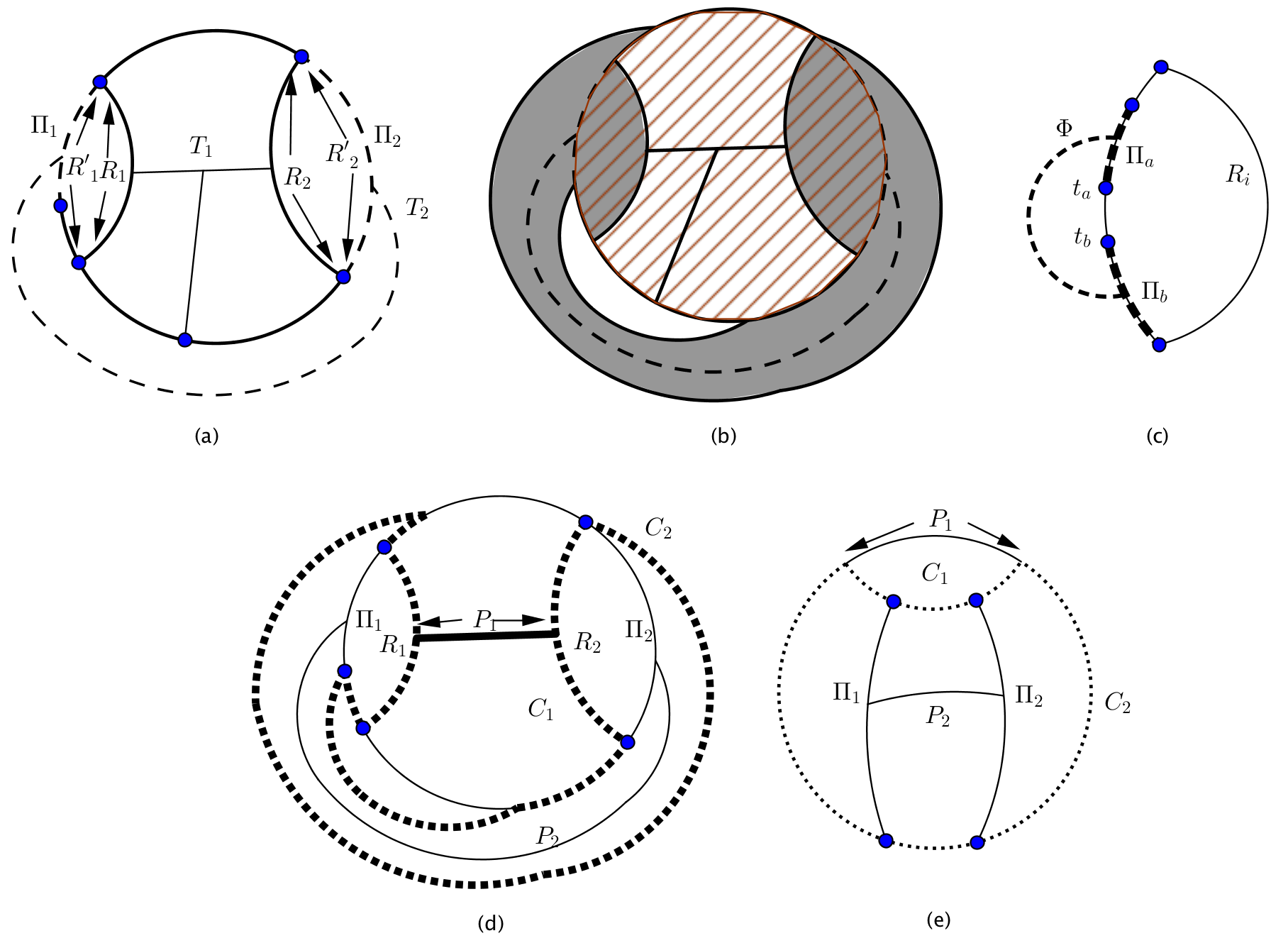}
\caption{
(a) The dashed tree is $T_2$ and there are two paths $\Pi_1$ and $\Pi_2$ of $T_2$ on $C(T_1)$. The blue vertices are terminals.
(b) The two shaded cycles are $(T_1)$ and $C(T_2)$, which shares two regions enclosed by $R_i$ and $R'_i$ for $i=1,2$.
(c) The dashed subtree is in $T_2$. $R'_i$ contains two subpaths of $T_2$ and two terminals $t_a$ and $t_b$ are enclosed.
(d) The dotted cycle is $C(T_2)$. 
The bold cycles $C_1$ and $C_2$ share $P_1$.
(e) An example for Simplified $C_1$ and $C_2$ in (d). The dotted cycle is $C(T_2)$ and the outer cycle is $C_2$. $C_1$ and $C_2$ share $P_1$.
}\label{fig: simplecycle}
\end{figure}
\begin{observation}\label{obs: onepath}
Every $R'_i$ contains at most one maximal path of $T_2$.
\end{observation}
\begin{proof}
For a contradiction, assume there is a path $R'_i$ that contains more than one maximal path of $T_2$. 
Let $\Pi_a$ and $\Pi_b$ be two successive such paths.
By the concept of terminal-bounded tree, there is a terminal-free path $\Phi$ of $T_2$ connecting $\Pi_a$ with $\Pi_b$ that only share endpoints with $C(T_1)$.
Refer to Figure~\ref{fig: simplecycle} (c).
So $\Phi$, $R_i$ and $R'_i$ (which contains two $R_i$-to-$\Phi$ paths) witness a cycle $C$ that strictly encloses an endpoint $t_a$ of $\Pi_a$ and an endpoint $t_b$ of $\Pi_b$.
Since the endpoints of $\Pi_a$ and $\Pi_b$ are leaves of $T_2$,
$t_a$ and $t_b$ should be in $C(T_2)$ by Property~(c).
So there is a subpath $P$ of $C(T_2)$ that is enclosed by $C$ such that $t_a\in P$.
Since the region enclosed by $C$ is divided into two parts by a subpath $\Phi'$ of $R'_i$ (which have the same endpoints as $\Phi$), and since one of the two parts is the face enclosed by $R_i$ and $R'_i$, we know
$P$ could only be in the region bounded by $\Phi'$ and $\Phi$.
However, $P$ can not cross $\Phi$ by Property~(a) since $\Phi$ is terminal-free and every vertex of $\Phi$ is an internal vertex of $T_2$; and $P$ can not cross any internal vertex of $\Phi'$ for otherwise $P$ will enter the face enclosed by $R_i$ and $R'_i$.
Therefore, $P$ and $R_i$ can not be connected, contradicting $P$ is a subpath of $C(T_2)$.
\end{proof}

Next, we construct two cycles $C_1$ and $C_2$ that share a subpath. 
Refer to Figure~\ref{fig: simplecycle} (d).
Since $T_1$ is a tree and $R_1$ and $R_2$ are vertex disjoint by Observation~\ref{obs: disjoint}, there is an $R_1$-to-$R_2$ subpath $P_1$ in $T_1$. 
Since $C(T_2)$ is simple and contains $R_1$ and $R_2$, $C(T_1)\cup P_1$ contains two simple cycles $C_1$ and $C_2$ whose intersection is $P_1$.
W.l.o.g.~assume $C_1$ is enclosed by $C_2$.

Since $T_2$ is a tree and $\Pi_1$ and $\Pi_2$ are vertex disjoint, there is a $\Pi_1$-to-$\Pi_2$ path $P_2$ in $T_2$.
Note that $P_2$ is terminal-free,
since $P_2$ does not contain any leaf of $T_2$ and terminals in $T_2$ are leaves by Lemma~\ref{lem: leaf}.
Let $xy$ be an edge of $P_2$.
Then by Lemma~\ref{lem: deletable}, $xy$ is nonremovable. 
However, $H$ and $xy$ also satisfy Lemma~\ref{lem: disjoint} with $C_1$ and $C_2$, which shows $xy$ is removable, giving a contradiction.  
\begin{description}
\item [Condition 1.]
Note that $C(xy)$ is enclosed by the cycle $C$ consisting of $\Pi_1$, $\Pi_2$, and two $\Pi_1$-to-$\Pi_2$ subpaths of $C(T_2)$: one is of $C_1$ and the other is of $C_2$. Refer to Figure~\ref{fig: simplecycle} (d) and (e).
Since $C$ is disjoint with the common subpath $P_1$ of $C_1$ and $C_2$, $C(xy)$ is also disjoint with $P_1$.

Showing the remainder of the first condition of Lemma~\ref{lem: disjoint} is similar to that of Claim~\ref{clm: condition1} if we replace $R_1$ and $R_2$ with the two $\Pi_1$-to-$\Pi_2$ subpaths of $C_1$ and $C_2$. 
Note that we have a stronger version of Observation~\ref{obs: neighbor} here, since $T_2$ is a tree.

\item [Condition 2.] If any neighbor $z$ of $xy$ in $C(xy)$ is not in $C_1\cup C_2$, it will be a non-leaf vertex in $T_2$, since all leaves of $T_2$ are in $C(T_2) \subseteq C_1\cup C_2$.
Then $z$ is not a terminal by Lemma~\ref{lem: leaf}.
We find a subpath $\Phi$ in $T_2$ from $z$ to $C(T_2) \subseteq C_1\cup C_2$ such that $\Phi$ only shares $z$ with $C(xy)$. 
By triconnectivity, there are at least three disjoint paths from $z$ to $C(T_2)$.
Since $C(xy)$ contains two such paths and encloses the fewest faces, there is a path $\Phi$ from $z$ to $C(T_2)$ outside of $C(xy)$.
If $\Phi$ shares any vertex with $C(xy)$ other than $z$, 
then $\Phi$ and the $C_1$-to-$C_2$ path that contains $z$ witness a cycle in $T_2$, a contradiction.
\end{description}

This proves the Tree Cycle Theorem.

\fi
\ifFull

\subsubsection{Terminal-free paths}\label{sec:strictly-term-free}
Let $P$ be a terminal-free $a$-to-$b$ path of $H$ such that there exists a cycle that strictly encloses the internal vertices of $P$.  Then $a,b\in C(P)$, since $H$ is triconnected by Lemma~\ref{lem: triconnected}. 
 Let $P_1(P)$ and $P_2(P)$ be the two $a$-to-$b$ subpaths of $C(P)$.

\begin{lemma}\label{lem: clean} 
If for every edge $e \in P$ there is a separating set $S_e \in C(P)$, then all the vertices inside of $C(P)$ are in $P$.
\end{lemma}
\begin{proof}
For a contradiction, assume there is a vertex $u$ strictly inside of $C(P)$ that is not on $P$. There can not be more than one path from $u$ to $C(P)$ disjoint from $P$, otherwise there will be another cycle which encloses fewer faces than $C(P)$ (Figure~\ref{fig: examples} (a)). By Lemma~\ref{lem: triconnected}, $H$ is triconnected, so there are at least two disjoint paths $R_1$ and $R_2$ from $u$ to $v_1$ and $v_2$ on $P$ disjoint from $C(P)$ (Figure~\ref{fig: examples} (b)).  For an edge $e$ on $P$ between $v_1$ and $v_2$, every separating set for $e$ must include a vertex on the path $R_1\cup R_2$, however, this contradicts the assumption that there is a separating set for $e$ in $C(P)$.
\begin{figure}[ht]
\centering
\includegraphics[scale=0.7]{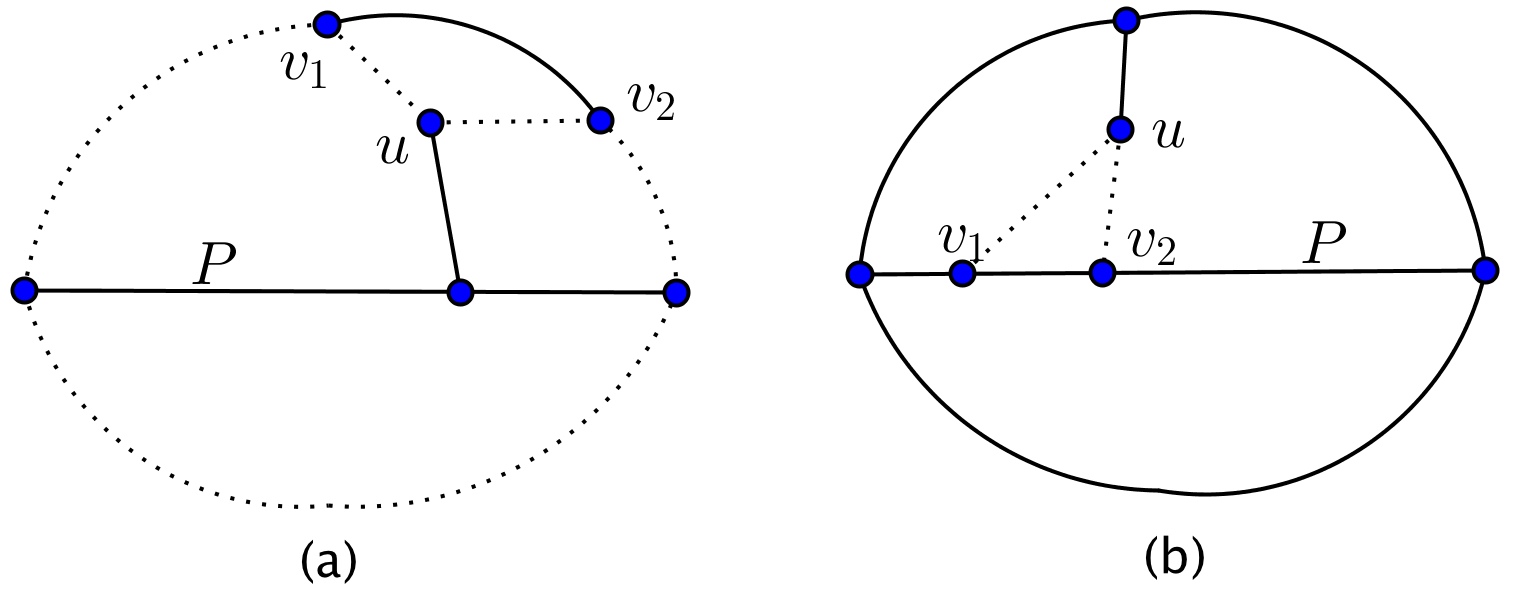}
\caption{(a) If vertex $u$ has two paths to $v_1$ and $v_2$ on $C(P)$ that are disjoint from $P$, then there is a smaller cycle (dotted) through $v_1, u, v_2$ that encloses $P$. (b) The separating set for an edge $e$ on $P$ between $v_1$ and $v_2$ must include a vertex of the path (dotted) strictly enclosed by $C(P)$, contradicting there is a separating set for $e$ in $C(P)$.}\label{fig: examples}
\end{figure}
\end{proof}

\begin{lemma}\label{lem: common} 
  No pair of adjacent vertices in $P$ has a common neighbor in $C(P)$.
\end{lemma}
\begin{proof}
  For a contradiction, assume there are adjacent vertices $u$ and $v$ with a common neighbor $z$ in $C(P)$.  Since $u,v,z$ forms a cycle, $S_{uv}$ must contain $z$.  Therefore, by Theorem~\ref{thm: removableone}, $uz$ is removable.  Further, both components of $H\setminus \{uv, S_{uv}\}$ must contain a vertex distinct from $u$ and $v$ and each of those vertices must have vertex disjoint paths to $S_{uv}$ and $uv$; therefore the degree of $z$ is at least 4.  Therefore, removing $uz$ will not result in contracting any edges incident to $z$ and so will preserve triconnectivity of the terminals, contradicting the minimality of $H$.
\end{proof}

For a separating pair $(e,S_e)$, let $\Sigma(e,S_e)$ be a closed curve that only intersects the drawing of $H$ in an interior point of $e$ and two vertices of $S_e$ and partitions the plane according to the components of $H\setminus (\{e\}\cup S_e)$.  Each portion of $\Sigma(e,S_e)\setminus H$ is contained in a face of $H$.  Since this is true for any $\Sigma(e,S_e)$, for two separating pairs $(e_1,S_{e_1})$ and $(e_2,S_{e_2})$ we may assume that the curves $\Sigma(e_1,S_{e_1})$ and $\Sigma(e_2,S_{e_2})$ are drawn so they cross each other at most 3 times.  $\Sigma(e_1,S_{e_1})$ and $\Sigma(e_2,S_{e_2})$ cross either at a point that is interior to a face of $H$ or at one vertex of $S_{e_1} \cup S_{e_2}$.  Since they are simple closed curves, they cross each other twice or not at all.

\begin{lemma}\label{lem: cut_on_cycle} 
For every edge $e$ in $P$, there exists a separating set for $e$ in $C(P)$.
\end{lemma}
\begin{proof}
We prove by induction on the subpaths of $P$: we assume the lemma is true for every strict subpath of $P$.  The base case is when $P$ is one edge $xy$: $S_{xy}$ must include one vertex of $P_1(xy)$ and one vertex of $P_2(xy)$ but does not include $x$ or $y$, giving the lemma.

Let $xy$ be any edge of $P$.  Without loss of generality, we assume $x \ne a$ and $y \notin P[a,x]$.
By the inductive hypothesis, there exists a separating set $S_{xy}=\{s_1, s_2\}$ in $C(P[x,b])$.  The following claim simplifies our proof, which we prove after using this claim to prove that $S_{xy} \in C(P)$.

\begin{claim}\label{clm: notonP}
  $S_{xy}$ does not contain an internal vertex of $P$.
\end{claim}

There are two cases:

\begin{enumerate}
\item If $b\in S_{xy}$, w.l.o.g.~assume $b=s_1$. Then we only need to
  show $s_2$ is in $C(P)$.  By the induction hypothesis, $s_2\in
  C(P[x,b])$. By Claim~\ref{clm: notonP}, $s_2$ can not be internal
  vertex of $P[a,y]$, so $a$ and $x$ must be in the same component of
  $H\setminus \{xy, S_{xy}\}$. Then $C(P[a,y])$ must contain $s_2$
  since $C(P[a,y])$ must intersect $\Sigma(xy, S_{xy})$ in vertices of
  $S_{xy}$.  Therefore, $s_2$ is in both of $C(P[a,y])$ and
  $C(P[x,b])$. By planarity, it must be in $C(P)$, for otherwise there
  will be other cycles which enclose fewer faces than $C(P[a,y])$
  and $C(P[x,b])$ and do not contain $s_2$.
\item If $b\notin S_{xy}$,
by Claim~\ref{clm: notonP}, $S_{xy}$ is not in $P$, so $a$ and $b$ are
in the distinct components of $H\setminus \{xy, S_{xy}\}$ since they
are connected to $x$ and $y$ by $P[a,x]$ and $P[y, b]$
respectively. That is, $b$ is strictly inside of $\Sigma(xy, S_{xy})$
and $a$ is outside of $\Sigma(xy, S_{xy})$. Then $C(P)$ must intersect
$\Sigma(xy, S_{xy})$ twice in vertices of $S_{xy}$. Because $C(P)$ is
simple, it can not intersect $\Sigma(xy, S_{xy})$ in the same vertex
of $S_{xy}$. Therefore, both vertices of $S_{xy}$ are in $C(P)$.
\end{enumerate}
This completes the proof of Lemma~\ref{lem: cut_on_cycle}.
\end{proof}
\begin{proof}[Proof of Claim~\ref{clm: notonP}]
 For a contradiction, assume $s_1$ is an internal vertex of $P$. Then it must in $P[a,x]$ since $S_{xy}$ is in $C(P[x,b])$. Further, $s_2$ can not be also in $P[a,x]$ since $x$ and $y$ are triconnected to $C(P)$, and then $C(P)$ and the paths from $x$ and $y$ to $C(P)$ contains a path from $x$ to $y$ disjoint from $P[a,x]$.

Since $a$ and $x$ are on different sides of $\Sigma(xy, S_{xy})$, $C(P[a,x])$ must cross $\Sigma(xy, S_{xy})$ twice. $\Sigma(xy, S_{xy})$ could only intersect $H$ at $xy$ and $S_{xy}$, so $C(P[a,x])$ must cross $\Sigma(xy, S_{xy})$ at $xy$ and $s_2$, for $s_1$ is in $P[a,x]$ and $C(P[a,x])$ is simple; w.l.o.g.~assume $P_1(P[a,x])$ contains $s_2$.

In $P[s_1, x]$, there must be a vertex between $s_1$ and $x$, for otherwise by Theorem~\ref{thm: removableone} edge $s_1x$ is removable, which contradicts Lemma~\ref{lem: deletable}.  Let $zx$ be an edge of $P[s_1, x]$. Then there exists a separating set $S_{zx}$ for $zx$ in $C(P[a, x])$ by induction hypothesis. 
We claim $\Sigma(zx, S_{zx})$ must intersect $P_1(P[a,x])$ between $s_2$ and $x$. If not, then it will intersect $P_1(P[a,x])$ outside of $\Sigma(xy, S_{xy})$. However, the $z$-to-$s_2$ path, together with the $s_2$-to-$x$ subpath of $P_1(P[a,x])$ and edge $zx$ form a cycle inside of $\Sigma(xy, S_{xy})$. 
See Figure~\ref{fig: sigma} (a).
Since $\Sigma(zx, S_{zx})$ must cross the edge $zx$, it must intersect the described cycle twice. Note that cycle is disjoint from $P_2(P[a,x])$. So $\Sigma(zx, S_{zx})$ will intersect the drawing of $H$ four times: twice at the described cycle inside of $\Sigma(xy, S_{xy})$, once at $P_1(P[a,x])$ outside of $\Sigma(xy, S_{xy})$ and once at $P_2(P[a,x])$. This contradicts the definition of $\Sigma(zx, S_{zx})$.  
Let $s$ be the vertex of $S_{zx}$ in $P_1(P[a,x])$. By the above argument, $s$ is between $s_2$ and $x$.
There are two cases.

\begin{enumerate}
\item  If $s = s_2$, the only vertex of $P_1(P[a,x])$ inside both of $\Sigma(zx, S_{zx})$ and $\Sigma(xy, S_{xy})$ is $s_2$. Then there is a path between $z$ and $s_2$ edge disjoint from $C(P[a,x])$ and $P$. By Lemma~\ref{lem: clean}, all vertices inside of $C(P[a,x])$ are in $P[a,x]$, so $z$ and $s_2$ are adjacent
and this edge is the only path between $z$ and $s_2$ disjoint from $P$ in $\Sigma(xy, S_{xy})$.  
Consider $C(P[z,b])$. It must intersect $\Sigma(xy, S_{xy})$ at $s_1$ and $s_2$. Then edge $zs_2$ must be in $C(P[z,b])$ since it is the only path between $z$ and $s_2$ disjoint from $P$ in $\Sigma(xy, S_{xy})$. By Lemma~\ref{lem: clean}, the $x$-to-$s_2$ path disjoint from $P$ in $\Sigma(xy, S_{xy})$ is an edge. However, the two edges $zs_2$ and $xs_2$ contradict Lemma~\ref{lem: common}.

\item If $s \ne s_2$, $S_{zx}$ must contain vertex $a$, for otherwise $\Sigma(zx, S_{zx})$ will intersect $P_1[a,x]$ between $a$ and $s_2$, which is the second intersection for $\Sigma(zx, S_{zx})$ and $P_1[a,x]$ and the fourth for $\Sigma(zx, S_{zx})$ and $H$. 
Then $s_1$ and $z$ are both connected to $s$ by paths edge disjoint from $C(P[a,x])$ and $P$ since they are inside of $\Sigma(zx, S_{zx})$. By Lemma~\ref{lem: clean} they are both adjacent to $s$. See Figure~\ref{fig: sigma} (c). However, this contradicts Lemma~\ref{lem: common}.
\qedhere
\end{enumerate}
\end{proof}
\begin{figure}[ht]
  \centering
  \includegraphics[scale=0.8]{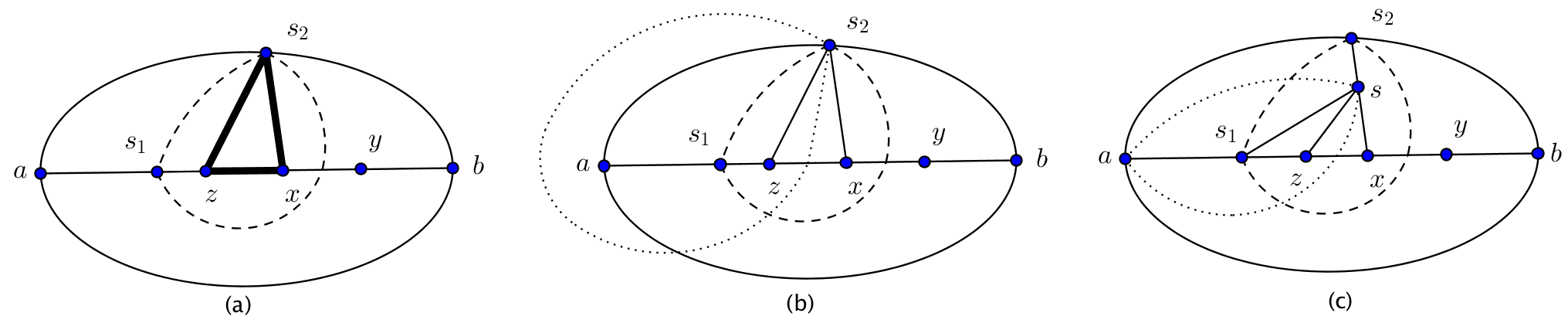}
  \caption{$\Sigma(xy,S_{xy})$ is the dashed cycle and $\Sigma(zx,S_{zx})$ is the dotted cycle.  (a) The bold cycle must be intersected by $\Sigma(zx, S_{zx})$ twice. (b) $\Sigma(zx, S_{zx})$ intersects the drawing of $H$ at $s_2$. (c) $\Sigma(zx, S_{zx})$ intersects the drawing of $H$ at $s$ and $a$.}
  \label{fig: sigma}
\end{figure}

\begin{lemma}\label{lem: outside}
Let $u_1$ and $u_2$ be the neighbors of an endpoint of $P$ on $C(P)$. Then there is a $u_1$-to-$u_2$ path whose internal vertices are strictly outside of $C(P)$.
\end{lemma}
\begin{proof} 
Without loss of generality, let $u_i$  be the neighbor of $a$ on $P_i(P)$, $i = 1,2$.  Let $\{s_1,s_2\}$ be the separating set for $ac \in P$ such that $s_i \in P_i(P)$, $i = 1,2$, as guaranteed by Lemma~\ref{lem: cut_on_cycle}.

If $u_i = s_i$, $i = 1,2$, then the component of $H \setminus \{ac,s_1,s_2\}$ that contains $a$ must contain another vertex $x$ and $x$ must have three vertex-disjoint paths to $a$, $s_1$ and $s_2$.  The latter two of these witness the $u_1$-to-$u_2$ path that gives the lemma.

If $u_1 \ne s_1$ and $u_2 = s_2$, then the component of $H \setminus \{ac,s_1,s_2\}$ that contains $u_1$ must have three vertex-disjoint paths to $a$, $s_1$ and $s_2$ and the latter of these paths witness the $u_1$-to-$u_2$ path that gives the lemma.  The case $u_1 = s_1$ and $u_2 \ne s_2$ is symmetric.

\begin{figure}[ht]
  \centering
  \includegraphics[scale=1]{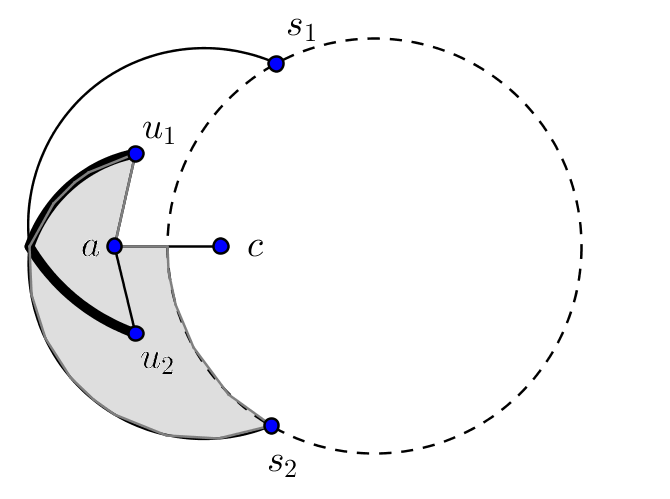}
  \caption{The construction of paths for Lemma~\ref{lem: outside}.  $\Sigma(ac,\{s_1,s_2\})$ is the dashed cycle.  The $u_2$ side of the cycle that separates $u_2$ and $s_1$ is shaded.  The path witnessing the lemma is bold and is formed by the $u_2$-to-$s_1$ and $u_1$-to-$s_2$ paths that avoid edges $u_1a$ and $u_2a$.} 
  \label{fig: outside}
\end{figure}

If $u_i \ne s_i$, $i = 1,2$, then the component of $H \setminus \{ac,s_1,s_2\}$ that contains $u_1$ and $u_2$ must have three vertex-disjoint paths from $u_i$ to $a$, $s_1$ and $s_2$ for $i = 1,2$.  The first of these, we may assume is the edge $u_ia$.  Consider the $u_1$-to-$s_2$ path; together with the edges $u_1a$ and $ac$ and the $s_2$-to-$ac$ portion of $\Sigma(ac,\{s_1,s_2\})$ that does not contain $s_1$, these form a closed curve in the plane that separates $u_2$ and $s_1$ (see Figure~\ref{fig: outside}).  The above-described $u_2$-to-$s_1$ path must therefore cross the $u_1$-to-$s_2$ path; these paths witness the $u_1$-to-$u_2$ path that gives the lemma.
\end{proof}

\begin{lemma}\label{lem: exist1}
For any terminal-free path $P$, there is a drawing of $H$ such that all internal vertices of $P$ are strictly enclosed by a simple cycle.
\end{lemma}
\begin{proof} 
For a contradiction, assume there is a terminal-free path $P$ whose internal vertices is not strictly enclosed by any cycle for any choice of infinite face of $H$: that is, every face of $H$ contains an internal vertex of $P$. Let $P$ be a minimal such path, let $a$ and $b$ be $P$'s endpoints, and let $c$ be $a$'s neighbor on $P$. Note that $b \ne c$, for otherwise $P$ is an edge and $H$ would have at most two faces (each containing $P$), but every triconnected graph has at least three faces.

First observe that there is a face $f$ whose bounding cycle strictly encloses all internal vertices of $P[c,b]$.  Let $f_1$ and $f_2$ be the two faces that contains edge $ac$: $f$ is one of $f_1$ and $f_2$. One of $f_1$ and $f_2$ only contains one internal vertex, namely $c$, of $P$, for otherwise both faces would contain at least two internal vertices of $P$ and every face of $H$ would contain an internal vertex of $P[c,b]$, contradicting the minimality of $P$.

Take $f$, defined in the previous paragraph, to be the infinite face of $H$. Since $f$'s bounding cycle strictly encloses the internal vertices of $P[c,b]$, $C(P[c,b])$ exists. Let $u$ and $v$ be $c$'s neighbors in $C(P[c,b])$. Note that $ac$ may or may not be in $C(P[c,b])$. By Lemma~\ref{lem: outside}, there is a $u$-to-$v$ path $R$ whose internal vertices are strictly outside of $C(P[c,b])$, so there is a drawing of  $H$ so that $R \cup C(P[c,b]) \setminus \{cu, cv\}$ is a cycle that strictly encloses internal vertices of $P$, a contradiction.  See Figure~\ref{fig: exist} (a). 
\begin{figure}[ht]
 \centering
 \includegraphics[scale=1.2]{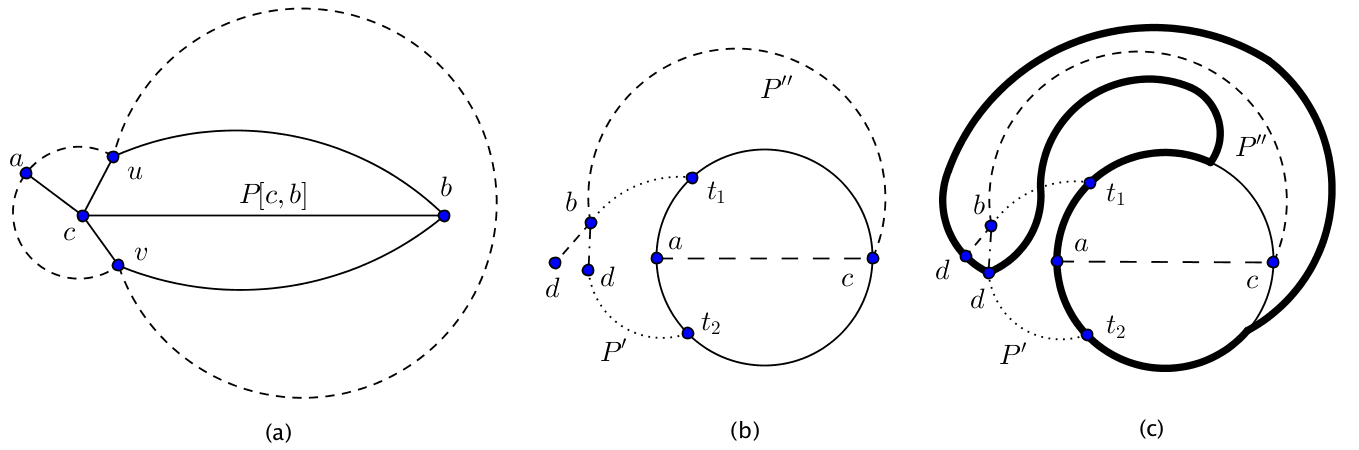}
 \caption{(a) The dashed paths show two possible $u$-to-$v$ paths $R$ outside of $C(P[c,b])$. The cycle $R \cup C(P[c,b]) \setminus \{cu, cv\}$ strictly encloses a face that does not contain any internal vertex of $P$ for an appropriate choice of infinite face. (b) The dotted path is $P'$ and dashed path is $P''$. Note that $d$ may or may not be in $P'$. (c) The bold cycle is $C(P'')$ and it crosses $P'$.}\label{fig: exist}
 \end{figure}
\end{proof}

\begin{proof} [Proof of Lemma~\ref{lem: exist}]
If $T^*$ is an edge, then this claim is trivial.  Suppose otherwise.

For a contradiction, assume every face of $H$ contains an internal vertex of $T^*$.
Let $ac$ be a leaf edge of $T^*$, where $a$ is a leaf of $T^*$. Recall that if a terminal-bounded component is not an edge, then it is obtained from a maximal terminal-free tree.  By the maximality of $T^*$, the two neighbors, $t_1$ and $t_2$, of $a$ on $C(ac)$ are terminals. By Lemma~\ref{lem: outside}, there exists a $t_1$-to-$t_2$ path $P'$ whose internal vertices are strictly outside of $C(ac)$. Choose $P'$ such that $C' = P' \cup \{at_1, at_2\}$ encloses the fewest faces. Then $C'$ does not strictly enclose any vertex since $H$ is triconnected by Lemma~\ref{lem: triconnected}.
By the assumption for the contradiction and the choice of $P'$, $P'$  must contain an internal vertex $b$ of $T^*$.

Let $P^*$ be the $a$-to-$d$ path of $T^*$ containing vertex $b$. Note that $d$ may or may not be in $P'$ (see Figure~\ref{fig: exist} (b)). By Lemma~\ref{lem: exist1}, $C(P^*)$ exists and by Lemmas~\ref{lem: clean} and~\ref{lem: cut_on_cycle}, $C(P^*)$ strictly encloses only internal vertices of $P^*$. Then $C'$ and $C(P^*)$ must cross each other and there exists a subpath of $C(P^*)$ that is strictly enclosed by $C'$, which contradicts that $C'$ encloses fewest faces.
See Figure~\ref{fig: exist} (c). 
\end{proof}

\begin{proof}[Proof of Lemma~\ref{lem: cycle}]
By Lemma~\ref{lem: exist1}, $C(P)$ exists. Then $P_1(P)$ and $P_2(P)$ both contain a terminal, for otherwise there is a cycle composed by $P$ and one of $P_1(P)$ and $P_2(P)$ that does not contain any terminal, contradicting Theorem~\ref{thm: minimal}.
We first prove that at least one of $P_1(P)$ and $P_2(P)$ contains more than one internal vertex, and then construct the cycle strictly enclosing $P$.
By Lemma~\ref{lem: clean} and~\ref{lem: cut_on_cycle}, all vertices inside of $C(P)$ are in $P$. If $P_1(P)$ and $P_2(P)$ both only contain one internal vertex, then the endpoints of the first edge of $P$ are both adjacent to the internal vertex of $P_1(P)$ or $P_2(P)$, which contradicts Lemma~\ref{lem: common}. Let
$x_a$ and $y_a$ (or $x_b$ and $y_b$) be the neighbors of $a$ (or $b$) in $P_1(P)$ and $P_2(P)$ respectively. Then at least one of $\{x_a, x_b\}$ and $\{y_a, y_b\}$ contains two distinct vertices.

By Lemma~\ref{lem: outside}, there is an $x_a$-to-$y_a$ path whose internal vertices are strictly outside of $C(P)$. We choose such an $x_a$-to-$y_a$ path $R_a$ such that the cycle $R_a\cup\{ax_a, ay_a\}$ encloses fewest faces. Then this cycle does not strictly enclose any vertex, for otherwise the vertex strictly inside of the cycle is triconnected to the cycle and we can find another cycle through that vertex which could  enclose fewer faces, contradicting the choice of $R_a$. Similarly, for $x_b$ and $y_b$, we can find an $x_b$-to-$y_b$ path $R_b$ such that $R_b\cup \{bx_b, by_b\}$ does not strictly enclose any vertex. Since at least one of $\{x_a, x_b\}$ and $\{y_a, y_b\}$ contains two distinct vertices, $R_1$ and $R_2$ are distinct. So the cycle $(R_a\cup R_b\cup C(P)) \setminus \{ax_a, bx_b, ay_a, by_b\}$ strictly encloses all the vertices of $P$ and only the vertices of $P$.
\end{proof}

\begin{proof} [Proof of Lemma~\ref{lem: disjoint}]
For a contradiction, assume edge $xy$ is nonremovable and consider $\Sigma(xy,S_{xy})$.  

First note that $\Sigma(xy,S_{xy})$ must be enclosed by $C_2$ and not enclosed by $C_1$  for otherwise, $\Sigma(xy,S_{xy})$ would intersect $C_1$ and $C_2$ in more than one point, and since $C_1$ and $C_2$ share at most one common subpath that is vertex disjoint from $C(xy)$ (by condition of the lemma), this would result in $|S_{xy}| \ge 3$, a contradiction.
Therefore $S_{xy}$ contains, w.l.o.g., two vertices of the $C_1$-to-$C_2$ path through $C(xy)$;  let $a$ be the vertex of $S_{xy}$ on the $x$-to-$C_1$ path. Refer to Figure~\ref{fig: removeedge}.  

Next note that $a$ must be a neighbor of $x$. For otherwise the neighbor $z$ of $x$ must be on the $x$ side of   $\Sigma(xy,S_{xy})$.  By condition of the lemma, there is a path from $z$ to $C_1$ or $C_2$ that is disjoint from $C(xy)$, however, the only way to cross $\Sigma(xy,S_{xy})$ is via a vertex of $C(xy)$, a contradiction.  Therefore $a$ is a neighbor of $x$.

Then $a$ has degree 2 in the part of $H$ on the $x$ side of $\Sigma(xy,S_{xy})$: one degree is given by the edge $ax$ and the other is given by the existence of a vertex $v \ne x$ on the $x$ side of $\Sigma(xy,S_{xy})$ which has vertex disjoint paths to $x$ and each vertex in $S_{xy}$.  For the same reason, $x$ has degree 4: degree 2 via the $C_1$-to-$C_2$ path, degree 1 via $y$ and degree 1 via the $v$-to-$x$ path.
    
Further we will argue that $a$ has degree 2 on the $y$ side of $\Sigma(xy,S_{xy})$ as well; $a$ then has degree 4.  By Theorem~\ref{thm: removableone}, $xa$ is removable.  Since $x$ and $a$ both have degree 4, removing $xa$ will not result in any edge contractions; this maintains the triconnectivity of the terminals, and contradicts the minimality of $H$.
    
        To show that $a$ has degree 2 on the $y$ side of  $\Sigma(xy,S_{xy})$, we have two cases.  If $a \in C_1$, then this follows from the two edges of $C_1$ incident to $a$.  If $a \notin C_1$, then by condition of the lemma, there is a path from $a$ to $C_1$ or $C_2$ that is disjoint from $C(xy)$ and so must be on the $y$ side of   $\Sigma(xy,S_{xy})$ and is notably disjoint from the 3 edges incident to $a$ on the $C_1$-to-$C_2$ path of $C(xy)$ and on the $v$-to-$a$ path.
\end{proof}
\begin{figure}  [ht]
    \centering
	\includegraphics[scale=0.9]{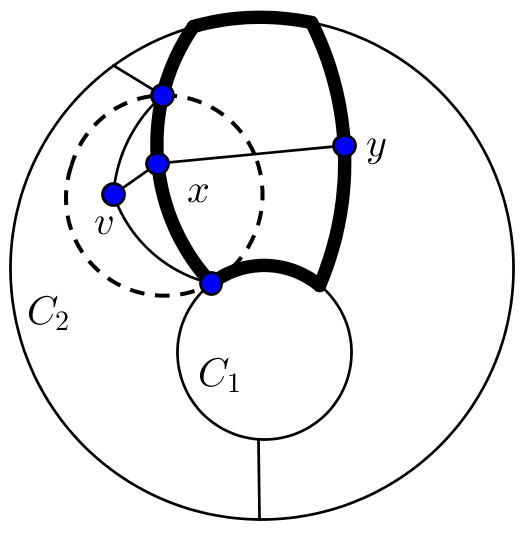}    
        \caption{The bold cycle is $C(xy)$ and the dashed cycle is $\Sigma(xy, S_{xy})$.} 
        \label{fig: removeedge}
    \end{figure}
\fi

\ifFull
\section{Correctness of spanner}\label{sec:spanner}
In this section, we prove the correctness of our spanner.
Let $\OPT$ be the weight of an optimal solution for 3-ECP.  
Then the correctness requires two parts: (1) bounding its weight by $O(\OPT)$ and (2) showing it contains a $(1+\epsilon)$-approximation of the optimal solution. The weight of our spanner is bounded by the weight of mortar graph, which we briefly introduce in Subsection~\ref{sec:mortargraph}.

The following Structure Theorem guarantees that there is a nearly-optimal solution in our spanner and completes the correctness of our spanner:
\begin{theorem}[Structure Theorem]
For any $\epsilon>0$ and any planar graph instance $(G,w, r)$ of 3-ECP, there exists a feasible solution $S$ in our spanner such that 
\begin{itemize}
\item the weight of $S$ is at most $(1+c\epsilon) \OPT$ where $c$ is an absolute constant, and 
\item the intersection of $S$ with the interior of any brick is a set of trees whose leaves are on the boundary of the brick and each tree has a number of leaves depending only on $\epsilon$.
\end{itemize}
\end{theorem}

We prove the Structure Theorem in Subsection~\ref{sec:structure-thm}. The idea is similar to that for 2-ECP, that is we transform an optimal solution to a feasible solution satisfying the theorem.  Throughout we indicate where the transformation for 3-ECP departs from those of 2-ECP. 
In the following, we denote by $Q$ the set of terminals, which are the vertices with positive requirement. 

\subsection{Mortar graph, bricks and portals}\label{sec:mortargraph} 
\fi
\ifFull
\else
\section{Proof of the Structure Theorem}
\ifFull In this section, we prove Structure Theorem.
\else In this section, we give a brief overview of the proof of the Structure Theorem (Theorem~\ref{thm:struct}); full details are in the appended full paper.
\fi First we introduce some properties of the mortar graph and bricks.
For a brick $B$, let $\partial B$ be its boundary and $\inter (B) = E(B) \setminus E(\partial B)$ be its interior. A  path is $\epsilon$-{\em short} if the distance between every pair of vertices on that path is at most $(1+\epsilon)$ times the distance between them in $G$. 
Bricks have the following properties.
\begin{lemma}
{\rm (Lemma 6.10~\cite{BKM09} rewritten)} The boundary of a brick $B$, in counterclockwise order, is the concatenation of four paths $W_B$, $S_B$, $E_B$ and $N_B$ (west, south, east and north) such that:
\begin{enumerate}
\item Every vertex of $Q\cap B$ is in $N_B \cup S_B$.
\item $N_B$ is 0-short and every proper subpath of $S_B$ is $\epsilon$-short.
\end{enumerate}
\end{lemma}
The paths that form eastern and western boundaries of bricks are called {\em supercolumns}, and the weight of all edges in supercolumns is at most $\epsilon \OPT$ (Lemma 6.6~\cite{BKM09}).
\fi
\ifFull
The {\em Mortar graph} is a grid-like subgraph of $G$ that (1) spans $Q$ and (2) has weight at most $9\epsilon^{-1}$ times the weight of a minimum  Steiner tree that spans $Q$. Since the weight of minimum Steiner tree is no more than OPT, the weight of mortar graph is no more than $9\epsilon^{-1}\OPT$. 
A brick $B$ is the subgraph of $G$ that is enclosed by a face of the mortar graph (including the boundary of the face); it has boundary $\partial B$ and interior $\inter(B) = G[E(B)\setminus E(\partial B)]$. 
Further, bricks have the following property:

\begin{lemma}
{\rm (Lemma 6.10~\cite{BKM09})} The boundary of a brick $B$, in counterclockwise order, is the concatenation of four paths $W_B$, $S_B$, $E_B$ and $N_B$ (west, south, east and north) such that:
\begin{enumerate}
\item The set of edges $B\setminus\partial B$ is nonempty.
\item Every vertex of $Q\cap B$ is in $N_B$ and $S_B$.
\item $N_B$ is 0-short and every proper subpath of $S_B$ is $\epsilon$-short.
\item There exists a number $t\le\kappa(\epsilon)$ and vertices $s_0,s_1,\dots,s_t$ ordered from west to east along $S_B$ such that for any vertex $x$ of $S_B[s_i,s_{i+1})$, the distance from $x$ to $s_i$ along $S_B$ is less than $\epsilon$ times the distance from $x$ to $N_B$ in $B$. 
\end{enumerate}
\end{lemma}
For the above lemma, a path is $\epsilon$-{\em short} if the distance between every pair of vertices on that path is at most $(1+\epsilon)$ times the distance between them in $G$ and $\kappa(\epsilon) = 4\epsilon^{-2}(1+\epsilon^{-1})$. The paths that forms eastern and western boundaries of bricks are called {\em supercolumns}, and further satisfy:
\begin{lemma}\label{lem: supercolumns}
{\rm (Lemma 6.6~\cite{BKM09})} The sum of the weight of all edges in supercolumns is at most $\epsilon\OPT$.
\end{lemma}
To obtain the spanner, we add a set of Steiner trees in each brick $B$ whose terminals are vertices of $\partial B$.  
The terminals are drawn from a subset of {\em portal} vertices evenly spaced on the boundary of each brick. 
\else
We designate a set of vertices, called {\em portals}, evenly spaced on the boundary of each brick. Each brick has only constant number (depending on $\epsilon$) of portals on its boundary.
\fi 
\ifFull
We bound the number of portals per brick by $\theta(\epsilon)=O(\epsilon^{-2}\alpha(\epsilon))$ and $\alpha(\epsilon)$ in turn depends on the number of connections required to allow a nearly optimal solution, which is bounded by $o(\epsilon^{-5.5})$.  
The portals satisfy:
\begin{lemma}
{\rm (Lemma 7.1~\cite{BKM09})} For any vertex $x$ on $\partial B$, there is a portal $y$ such that the weight of $x$-to-$y$ subpath of $\partial B$ is at most $1/\theta(\epsilon)$ times the weight of $\partial B$. 
\end{lemma}

Since the weight of each Steiner tree in a brick $B$ can be bounded by the weight of the weight of $\partial B$, and since there are only constant number (that is $2^{\theta(\epsilon)}$) of such Steiner trees, the weight of all trees we add in $B$ is at most $f(\epsilon)$ times the weight of $\partial B$. 
So the total weight of our spanner is bounded by $(2f(\epsilon) + 9\epsilon^{-1})$ times the weight of $MG$, which is $O(\OPT)$.

\subsection{Proof of the Structure Theorem}\label{sec:structure-thm}

We 
\else

To prove the Structure Theorem, we 
\fi
transform $\OPT$ for the instance $(G, Q, r)$ so that it satisfies the following properties (repeated from the introduction):
\begin{description}
\item [P1:] $\OPT \cap \inter(B)$ can be partitioned into a set of trees $\mathcal T$ whose leaves are on the boundary of $B$.
\item [P2:] If we replace any tree in $\mathcal T$ with another tree spanning the same leaves, the result is a feasible solution.
\item [P3:] There is another set of $O(1)$ trees $\mathcal T'$ that costs at most a $1+\epsilon$ factor more than $\mathcal T$, such that each tree of $\mathcal T'$ has $O(1)$ leaves and $(\OPT \setminus {\mathcal T} ) \cup {\mathcal T}'$ is a feasible solution.
\end{description}
\ifFull To argue about the leaves of trees on the boundary of bricks, we use the following definition:
    \begin{definition}
{\rm(Joining vertex \cite{BK13})}. Let $H$ be a subgraph of $G$ and $P$ be a subpath of $\partial G$. A joining vertex of $H$ with $P$ is a vertex of $P$ that is the endpoint of an edge of $H-P$.
\end{definition}
\fi
The transformation consists of the following steps:

        \begin{description} 
        
        \item[Augment] We add four copies of each supercolumn; we take two copies each to be interior to the two adjacent bricks. After this, connectivity between the east and west boundaries of a brick will be transformed to that between the north and south boundaries.
\ifFull
By Lemma~\ref{lem: supercolumns}, 
\else
Since the weight of all supercolumns is at most $\epsilon \OPT$,
\fi
this only increases the weight by an small fraction of $\OPT$.


    \item[Cleave] By cleaving a vertex, we split it into multiple
     copies while keeping the connectivity as required by adding
     artificial edges of weight zero between two copies and maintaining a planar
     embedding.
     We call the resulting solution $\OPT_C$.
In this step, we turn $k$-edge-connectivity into $k$-vertex-connectivity for $k =1,2,3$. By Theorem~\ref{thm: minimal}, we can obtain Property P1: $\OPT_C \cap \inter(B)$ can be partitioned into a set $\cal T$ of trees whose leaves are in $\partial B$.
By Corollary~\ref{cor: terminal}, we can obtain Property P2: we can obtain another feasible solution by replacing any tree in $\cal T$ with another tree spanning the same leaves.

    \item[Flatten] For each brick $B$, we consider the connected
     components of $\OPT_C \cap \inter(B)$.
     If the component only spans vertices in the north or south
     boundary, we replace it with the minimum subpath of the boundary
     that spans the same vertices. 
This will not increase the weight much by the $\epsilon$-shortness of the north and south boundaries. Note that vertex-connectivity may bread as a result, but edge-connectivity is maintained.
In the remainder, we only maintain edge-connectivity.
We call the resulting solution $\OPT_F$.


\item[Restructure] For each brick $B$, we consider the connected
     components of $\OPT_F \cap \inter(B)$. 
     We replace each component with a subgraph through a mapping
     $\phi$. The new subgraph may be a tree or a subgraph
     $\widehat{C}$ \ifFull given by Lemma~\ref{lem: newmap}. \else that is the union of a cycle and two subpaths of $\partial B$. \fi  The mapping
     $\phi$ has the following properties:
    \begin{enumerate}
    \item For any component $\chi$ of $\OPT_M \cap \inter(B)$,
     $\phi(\chi)$ is connected and spans $\chi \cap \partial B$.
    \item For two components $\chi_1$ and $\chi_2$ of $\OPT_M \cap
     \inter(B)$, if $\phi(\chi_i) \ne \widehat{C}$ for at least one of
     $i=1,2$, then $\phi(\chi_1)$ and $\phi(\chi_2)$ are edge-disjoint,
     taking into account edge multiplicities.
    \item The new subgraph $\phi($OPT$_M \cap \inter(B) )$ has 
     \ifFull
     $\alpha(\epsilon) = o(\epsilon^{-5.5})$ joining vertices with
     $\partial B$.
     \else only constant number (depending on $\epsilon$) of vertices in the boundary $\partial B$.
     \fi
    \end{enumerate}
We can prove that the total weight is increased by at most $\epsilon \OPT_F$, giving Property P3.
    We call the resulting solution $\OPT_R$.

    \item[Redirect] \ifFull We connect each joining vertex $j$ of $\OPT_R \cap
     \inter(B)$ 
\else
	We connect each vertex $j$ of $\OPT_R \cap \inter(B)$ in the boundary $\partial B$
\fi
to the nearest portal $p$ on $\partial B$ by adding multiple copies of the short $j$-to-$p$ subpath of $\partial B$. 
\ifFull
     We call the resulting solution $\widehat{OPT}$.
\else
Similar to 2-ECP, we can prove this only increases the weight by an $\epsilon$ fraction of $\OPT$ and the resulting solution satisfies the Structure Theorem.
\fi
    \end{description}


\ifFull
We give more details of these transformations in the following subsections and argue that the transformations guarantee the required connectivity as we do so.  These details are very similar to that used for 2-ECP; we note the differences.  The Restructure step which requires a different structural lemma than used for Borradaile and Klein's PTAS for 2-ECP~\cite{BK08}; the difference between 
 Lemma~\ref{lem: newmap} and that used by Borradaile and Klein is that $\widehat{C}$ need only be a cycle for the 2-ECP, whereas to maintain 3-edge connectivity, a more complicated subgraph $\widehat{C}$ is required.
We borrow the following lemma from~\cite{BK08} to prove Lemma~\ref{lem: newmap}.
\begin{lemma}\label{lem: simpleoneroot}
{\rm(Simplifying a tree with one root, Lemma 10.4 \cite{BKM09}). }Let $r$ be a vertex of $T$. There is another tree $\widehat{T}$ that spans $r$ and the vertices of $T\bigcap P$ such that $w(\widehat{T})\leq (1+4\epsilon)w(T)$ and $\widehat{T}$ has at most $11\epsilon^{-1.45}$ joining vertices with $P$.
\end{lemma}
\begin{lemma}\label{lem: newmap}
Let $\mathcal{F}$ be a set of non-crossing trees whose leaves are joining vertices with $\epsilon$-short paths $P_1$ and $P_2$ on the boundary of the graph, and each tree in $\mathcal{F}$ has leaves on both paths. Then there is a subgraph or empty set $\widehat{C}$, a set $\widehat{\mathcal{F}}$ of trees, and a mapping $\phi:\mathcal{F}\to\widehat{\mathcal{F}} \cup \widehat{C}$ with the following properties
\begin{itemize}
\item For every tree $T$ in $\mathcal{F}$, $\phi(T)$ spans $T$'s leaves.
\item For two trees $T_1$ and $T_2$ in $\mathcal{F}$, if $\phi(T_i) \neq \widehat{C}$ for at least one of $i=1,2$ then $\phi(T_1)$ and $\phi(T_2)$ are edge-disjoint (taking into account edge multiplicities).
\item The subgraph $\bigcup \widehat {\mathcal{F}} \cup\{ \widehat{C}\}$ has $o(\epsilon^{-2.5})$ joining vertices with $P_1\cup P_2$.
\item $w(\widehat{C})+ w(\widehat{\mathcal{F}})  \leq 6 w(P_2) + (1+d\cdot\epsilon) w(\mathcal{F})$, where $d$ is an absolute constant.
\item Any two vertices of $\widehat{C} \cap (P_1\cup P_2)$ are three-edge-connected.
\item For any three pairs of vertices of $\widehat{C} \cap (P_1 \cup P_2)$ (where vertices may be repeated), $\widehat{C}$ contains three edge-disjoint paths connecting them respectively. 
\end{itemize}
\end{lemma}
\begin{proof}
Call $P_1$ the top path and $P_2$ the bottom path.  Order the trees of $\mathcal F$: $T_1, T_2, \ldots, T_k$  according to their leaves on $P_2$ from left-to-right; the trees are well ordered since they are non-crossing and have leaves on both paths.
There are two cases:
\begin{description}
        \item [ $ k>\epsilon^{-1}$:] In this case we reduce the
     number of trees by incorporating a subgraph $\widehat{C}$. Let
     $a$ be the smallest index such that $w(T_a)\leq \epsilon
     w(\mathcal{F})$ and $b$ the largest index such that $w(T_b)\leq
     \epsilon w(\mathcal{F})$. We replace trees
     $T_a,T_{a+1},\dots,T_b$ with a subgraph $\widehat{C}$. Let $P'_1$
     be the minimal subpath of $P_1$ that spans all the leaves of tree
     $T_i$ for $a\leq i\leq b$. $P'_2$ is likewise the minimal subpath
     on $P_2$. Let $u_i$ and $v_i$ be $P'_i$'s endpoints (with $u_i$
     the left end) for $i=1,2$. Let $P_a$ ($P_b$) be the $u_1$-to-$u_2$
     ($v_1$-to-$v_2$) subpath in $T_a$ ($T_b$).  Since $P_1$ is
     $\epsilon$-short, we have $$w (P'_1) \leq (1+\epsilon) w (P_a
     \cup P_b \cup P'_2).$$ Let $\widehat{C}$ be the multisubgraph $\{P_a,P_a, P_b ,P'_1,P'_1,P'_2,P'_2\}$.  We replace $\bigcup_{i=a}^bT_i$
     with $\widehat{C}$ and set $\phi(T_i) = \widehat{C}$ for $a \leq i
     \leq b$.

    \item[$ k\leq \epsilon^{-1}$:] In this case, the number of trees is already bounded and we set  $\widehat{C} = \emptyset$.
    \end{description}
In both cases, we transform the remaining trees  $T_i$ ($i\neq a,a+1,\dots,b$) as follows. Let $T'_i$ be a minimal subtree of $T_i$ that spans all leaves of $T_i$ on $P_1$ and exactly one vertex $r$ on $P_2$. Let $R_i$ be the minimal subpath on $P_2$ that spans all leaves of $T_i$ on $P_2$. By Lemma \ref{lem: simpleoneroot}, there is another tree, say $T''_i$, for $T'_i$ with root $r$ and path $P_1$. We replace $T_i$ with $\widehat{T_i}=T''_i\cup R_i$, and set    $\phi(T_i)=\widehat{T_i}$ for $i\neq a,\dots,b$. Then $\widehat{T_i}$ spans all leaves of $T_i$.

$\widehat{C}$ (if non-empty) has six joining vertices with $P_1\cup P_2$. Each tree $\widehat{T_i}$ has one joining vertex with $P_2$ and by Lemma \ref{lem: simpleoneroot} $o(\epsilon^{-1.5})$ joining vertices with $P_1$. By the choice of $a$ and $b$, there are at most $\epsilon^{-1}$ trees mapped to $\widehat{T_i}$. So $\widehat{\mathcal{F}}\cup\{\widehat{C}\}$ has totally at most $o(\epsilon^{-2.5})$ joining vertices with $P_1\cup P_2$.

The total weight of $\widehat{C}$ is $$\begin{array}{llll}
w(\widehat{C})& \leq 2 w(P'_1)+2 w(P'_2)+ 2 w(P_a)+ w(P_b) \\
&\leq 2(1+\epsilon) [w(P_a)+ w(P_b)+ w(P'_2)] +2 w(P'_2)+ 2 w(P_a)+ w(P_b)\\
&\leq (4+2\epsilon) w(P_a)+ (3+2\epsilon) w(P_b)+ (4+2\epsilon) w(P'_2)\\ 
&\leq (7+4\epsilon) \epsilon \cdot w (\mathcal{F})+ (4+2\epsilon) w(P'_2).\\
\end{array}$$
And the total weight of $\widehat{T_i}$ is $$\begin{array}{llll}
\sum_{i=1,\dots,a-1,b+1,\dots,k}w(\widehat{T_i})&\leq\sum_{i=1,\dots,a-1,b+1,\dots,k} (w(R_i)+ (1+4\epsilon) w(T_i)).
\end{array}$$
Since all the trees in $\mathcal{F}$ are non-crossing, $R_i$ and $P'_2$ are disjoint. The total weight of our replacement is at most $6w(P_2)+(1+O(\epsilon))w(\mathcal{F})$.

Now we prove the connectivity properties for $\widehat{C}$. By its construction, $\widehat{C}$ spans all the leaves of tree $T_i$ for $a \leq i \leq b$  and vertices of $\widehat{C} \cap (P_1 \cup P_2)$ are three-edge-connected. So we only need to prove $\widehat{C}$ contains three edge-disjoint paths connecting any three pair of vertices of $\widehat{C} \cap (P_1 \cup P_2)$ respectively.  This can be seen by case analysis, as described and illustrated below.  Let $P^*=P_a \cup P'_1 \cup P'_2$, then every edge of $P^*$ has multiplicity of two in $\widehat{C}$ by construction.
 
We first consider the case that any two pairs do not contain identical vertex. For $i=1,2$ let $x_i$, $y_i$ and $z_i$ be the three pair of vertices.  If there are two pairs of vertices are not interleaving in $P^*$, then we could connect these two pairs by $P^*\cup P_b$, which is a cycle. And the other pair could be connected by a subpath of $P^*$ which is edge-disjoint from the other two paths in $\widehat{C}$. Otherwise, we have the two sets, $\{x_i, y_j, z_l\}$ and $\{x_{3-i}, y_{3-j}, z_{3-l}\}$, which do not interleave each other and appear in the same order in $P^*$. See Figure~\ref{fig: colorpaths}: there are three edge-disjoint paths connecting the three pairs respectively. 
\begin{figure}[ht]
  \centering
  \includegraphics[scale=1]{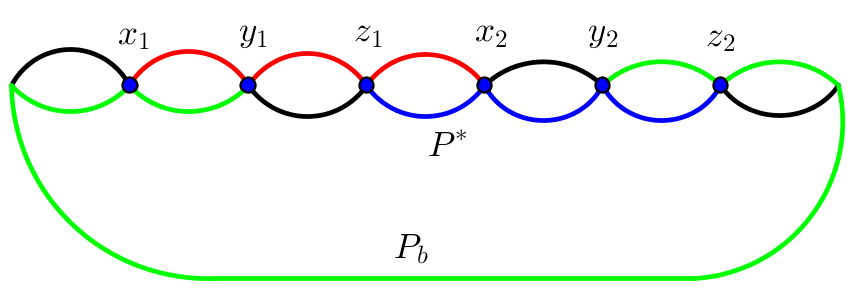}
  \caption{If every two pairs of the three interleave, there are three edge-disjoint paths, shown in different colors, connecting them respectively in $\widehat{C}$.}
  \label{fig: colorpaths}
\end{figure}

If there are two pairs containing identical vertices, we could connect these two pairs by $P^*\cup P_b$ and then connect the other pair by another subpath of $P^*$. Since every edge in $P^*$ is multiple in $\widehat{C}$, these paths are edge-disjoint.
\end{proof}

\subsubsection{Augment} 
For each supercolumn $P$ common to two bricks $B_1$ and $B_2$, we do the following:
\begin{itemize}
\item Add four copies of $P$, $P_1, P_2, P_3, P_4$ to $\OPT$.  We consider $P_1$ and $P_2$ to be internal to $B_1$, and $P_3$ and $P_4$ to be internal to $B_2$.  We also cut open the graph along $P$, creating a new face between $P_2$ and $P_3$.  Call the result $\OPT'_A$.
\item Remove edges from $\OPT'_A$ until what remains a minimal subgraph satisfying the connectivity requirements.  Similar to the argument illustrated in Figure~\ref{fig: colorpaths}, maintaining connectivity in the presence of this new face is achievable.
\end{itemize}

Let the resulting solution be $\OPT_A$. By this construction, $\OPT_A$ has no joining vertices with internal vertices of the supercolumns:

\begin{lemma}\label{lem: terminal_NS}
For every brick $B$, the joining vertices of OPT$_A\cap B$ with $\partial B$ belong to $N_B$ and $S_B$.
\end{lemma}

\subsubsection{Cleave}\label{sec:cleave}

In this part, we call the vertices with positive requirement terminals.  Given a vertex $v$ with a non-interleaving bipartition $A$ and $B$ of the edges incident to $v$, we define {\em cleaving} $v$ as the following: split $v$ into two copies $v_1$ and $v_2$, such that all the edges in $A$ ($B$) are incident to $v_1$ ($v_2$); and then we add one zero-weight edge between $v_1$ and $v_2$.  See Figure~\ref{fig: cleavings}.  We have two types of cleavings:

{\bf Simplifying cleavings} Let $C$ be a non-self-crossing, non-simple cycle that visits vertex $v$ twice. For $1\leq i \leq 4$, let $e_i$ be the edge of $C$ incident to $v$ such that $e_1$, $e_2$, $e_3$ and $e_4$ are embedded clockwise and $e_1$ and $e_4$ are in the same subcycle of $C$.  Define a bipartition $A$, $B$ of the edges incident to $v$ as follows: given the clockwise embedding of those edges, let $A$ contain $e_1$, $e_2$ and all the edges between them clockwise.

{\bf Lengthening cleavings} Let $C$ be a cycle and $v$ a vertex on $C$. Let $e_1$ and $e_2$ be the edges incident to $v$ strictly inside of $C$, and let $e'_1$ and $e'_2$ be the edges of $C$ incident to $v$. Define a bipartition $A$ and $B$ of the edges incident to $v$ as follows: given the embedding of those edges with $e_1$, $e_2$, $e'_2$ and $e'_1$ in the clockwise order, $A$ contains $e_1$, $e'_1$ and all the edges from $e'_1$ to $e_1$ clockwise. 

\begin{figure}[ht]
  \centering
  \includegraphics[scale=1]{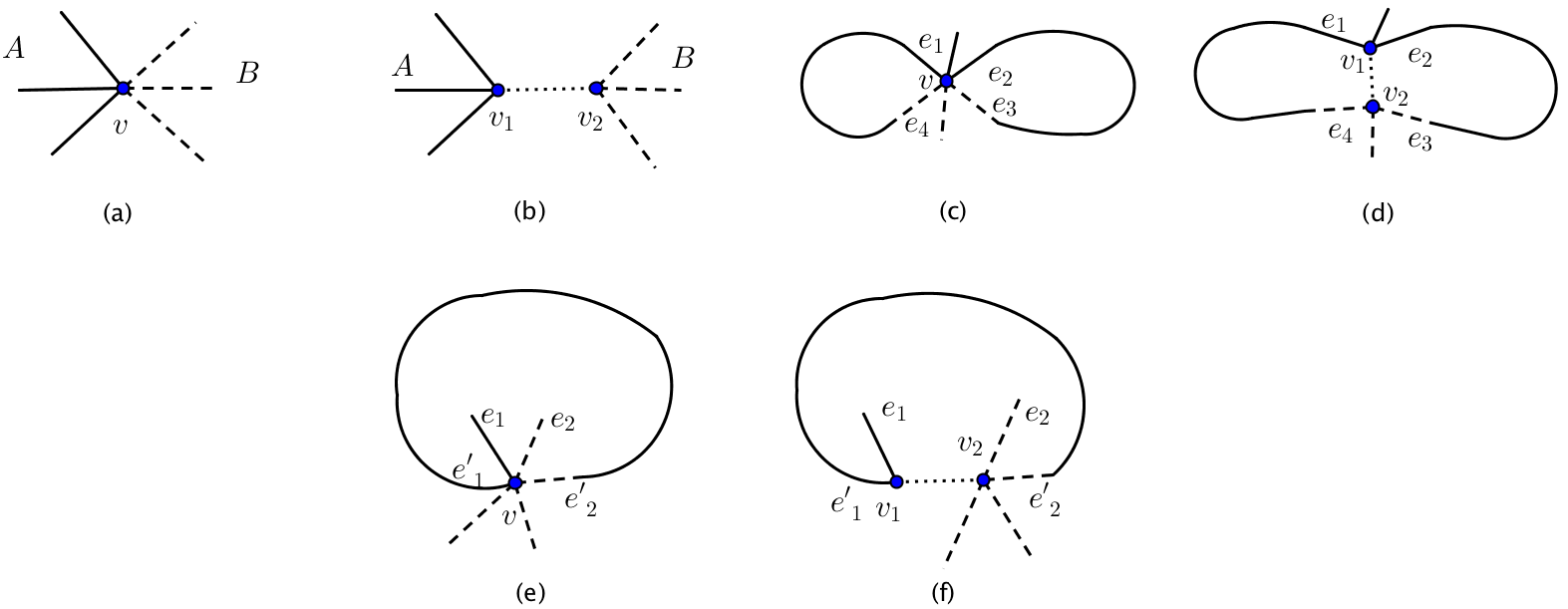}
  \caption{Cleaving examples. The bipartition of edges incident to $v$ is illustrated by solid edges $A$ and dashed edges $B$, and the added artificial edge is illustrated by the dotted edge. (a) and (b) give an example before and after a cleaving. (c) and (d) illustrate a simplifying cleaving. (e) and (f) show a lengthening cleaving. }
  \label{fig: cleavings}
\end{figure}

We perform simplifying cleavings for all the non-simple cycles of OPT$_A$ until every cycle is simple and call the resulting solution OPT$'_S$ and the resulting graph $G_S$. Note that after simplifying cleaving a terminal, we take only one copy that is in the mortar graph as the terminal.  These cleavings will reduce edge-connectivity to vertex-connectivity by the following lemmas.
\begin{lemma}\label{lem: copies}
For the cleaved vertex $v$, the copies of $v$ are three-edge-connected in OPT$'_S$. 
\end{lemma}
\begin{proof}
The original cycle will give two edge-disjoint paths between the two copies. And the artificial edge will be the third path.
\end{proof} 

\begin{lemma}
Let $\widehat{H}$ be the graph obtained from $H$ by a simplifying cleaving a vertex $v$. Then for $k=1,2,3$, if a vertex $u$ and $v$ are $k$-edge-connected in $H$, then $u$ and any copy of $v$ are $k$-edge-connected in $\widehat{H}$.
\end{lemma}
\begin{proof}
If $u$ and $v$ are $k$-edge-connected, then there are $k$ edge-disjoint paths from $u$ to $\{v_1, v_2\}$. Since $k$ is no more than three, and $v_1$ and $v_2$ are three-edge-connected by Lemma~\ref{lem: copies},  $u$ and $v_1$ are $k$-edge-connected. Similarly, $u$ and $v_2$ are $k$-edge-connected.
\end{proof}
\begin{corollary}
For $k=1,2,3$, if two vertices are $k$-edge-connected in OPT$_A$, then any of their copies are $k$-edge-connected in OPT$'_S$.
\end{corollary}

We then remove edges from OPT$'_S$ until the rest is a minimal subgraph satisfying the connectivity requirement. We call the resulting solution OPT$_S$.  Since every cycle in OPT$_S$ is simple, we have the following lemma:
\begin{lemma}\label{lem: simplecleaving}
For $k=1,2,3$, if two terminals are $k$-edge-connected in OPT$_A$, then they are $k$-vertex-connected OPT$_S$.
\end{lemma}

Let $Q_{23}$ be set of terminals whose requirement is at least two  and $H$ be the minimal $(Q_{23}, r)$-vertex-connected subgraph in $\OPT_S$.  Since every cycle in OPT$_S$ is simple, we know that $H$ also has this property. It follows that the degree of any vertex $v$ of $H$ is no more than three, since $H$ is biconnected by Lemma~\ref{lem: biconnected}:  for otherwise, let the first ear (a cycle) contain $v$ and let the second and third ear start with the next two edges incident to $v$ -- from these ears it is easy to construct two cycles that only meet at $v$, witnessing a non-simple cycle.

We perform lengthening cleavings w.r.t.\ the boundary of each brick and the edges of $\OPT_S$ that are incident to a boundary vertex in that brick until the vertices of the brick boundaries have at most one edge in the solution in the interior of each incident brick.  We add the introduced zero-length edges to the solution.  Call the new solution  $\OPT_C$ and the new graph and mortar graph $G_C$ and  $MG_C$, respectively.

 For each cleaved terminal, we assign only one copy to be the terminal: we pick the copy of highest degree in $\OPT_C$ to be the terminal; note that terminals are still in the mortar graph.
        Although the flatten, restructure and redirect steps may break the vertex connectivity guaranteed by Lemma~\ref{lem: copies}, the following lemmas will guarantee that edge-connectivity will not be preserved.
 
\begin{lemma}\label{lem: lengthen}
For $k=1,2,3$, if terminals are $k$-vertex-connected in $\OPT_S$, then they are $k$-vertex-connected $\OPT_C$.  Further, $\OPT_C$ is minimal.
\end{lemma}

\begin{proof}
Borradaile and Klein prove (Lemma~5.9~\cite{BK13}) that lengthening cleavings maintain biconnectivity, so we need only argue that terminals $a$ and $b$ that are  triconnected in $\OPT_S$ have terminal copies $a'$ and $b'$ that are triconnected in $\OPT_C$.

Consider three, vertex-disjoint $a$-to-$b$ paths $P_1, P_2, P_3$ in $\OPT_S$. As argued earlier, the degrees of $a$ and $b$ in $\OPT_S$ are three.  The terminal copy $a'$ of $a$ likewise has degree 3 in $\OPT_C$: if $a$ is lengthening cleaved, then there is a copy of $a$ that has degree 3.

The paths $P_1, P_2, P_3$ map to paths $P_1', P_2',P_3'$ between copies of $a$ and $b$ in $\OPT_C$: if an internal vertex of $P_i$ is lengthening cleaved, then include the introduced edge in the path.  The endpoints of $P_i'$ that map to copies of, w.l.o.g., $a$ are likewise connected by introduced edges (if $a$ is subject to a lengthening cleaving).  Let $a'$ be the terminal copy of $a$; we augment the paths $P_i'$ to connect to $a'$ via an introduced, zero-weight edge.  Doing so for $b'$ as well gives three vertex disjoint $a'$-to-$b'$ paths in $\OPT_C$.

That $\OPT_C$ is minimal follows directly from the fact that $\OPT_S$ is minimal.
\end{proof}

The following lemma was proved by Borradaile and Klein for 2-ECP (Lemma 5.10~\cite{BK13}); their proof only relied on the fact that every cycle contains a terminal.  Since this is true for the 3-ECP too (Theorem~\ref{thm: minimal}), we get the same lemma:

\begin{lemma}\label{lem: terminal_cycle} Let $B$ be a brick in $G_C$ with respect to $MG_C$. The intersection OPT$_C\cap \inter(B)$ is a forest whose joining vertices with $\partial B$ are all the leaves of the forest.
\end{lemma}

\begin{lemma}\label{lem: final}
Let $x$ and $y$ be two terminals in OPT$_C$ such that $k= \min\{r(x), r(y)\} \geq 2$ and let $B$ be a brick. There exist $k$ vertex-disjoint paths from $x$ to $y$ such that for any two such paths $P_1$ and $P_2$, any connected component of $P_1 \cap \inter(B)$ and any connected component of $P_2 \cap \inter(B)$ belong to distinct components of OPT$_C \cap \inter(B)$.
\end{lemma}

\begin{proof}
Let $H_C$ be the minimal $(Q_{23},r)$-vertex-connected subgraph in OPT$_C$. Then each connected component of OPT$_C\setminus H_C$ is a tree and is connected to $H_C$ by one edge.
By Corollary~\ref{cor: terminal}, there are $min\{ r(x), r(y)\}$ vertex-disjoint paths from $x$ to $y$ in $H_C$ (and then in OPT$_C$) such that any path connecting any two of those $x$-to-$y$ paths contains a terminal. Let $P_1$ and $P_2$ be any two such $x$-to-$y$ paths. 
Since all terminals are on the boundaries of bricks, any $P_1$-to-$P_2$ path in $\inter(B)$ will be divided into two subpaths by some lengthening cleaving. 
So any connected component of OPT$_C \cap \inter(B)$ can not contain the components of both $P_1 \cap \inter(B)$ and $P_2 \cap \inter(B)$.
\end{proof}

\subsubsection{Flatten}

This step is the same as described by Borradaile and Klein for 2-ECP~\cite{BK13}. For each brick $B$, consider the edges of OPT$_C\cap \inter(B)$. By Lemma \ref{lem: terminal_cycle}, the connected components of OPT$_C \cap \inter(B)$ are trees. By Lemma \ref{lem: terminal_NS}, every leaf is either on $N_B$ or $S_B$. For every tree whose leaves are all on $N_B$ ($S_B$), we replace the tree with the minimal subpath of $N_B$ ($S_B$) that contains all its leaves. Let the resulting solution be OPT$_F$. By Lemma~\ref{lem: final}, this guarantees 2- and 3-edge-connectivity between terminals as required; trees may be flattened against a common $\epsilon$-short path that is the northern boundary of one brick and the southern boundary of another, so vertex-connectivity gets broken at this stage.

\subsubsection{Restructure}  
This step is the same as described by Borradaile and Klein for 2-ECP~\cite{BK13}, except we apply our 3-ECP specific lemma (Lemma~\ref{lem: newmap}).  Restructuring replaces  $\OPT_F \cap \inter(B)$, which is a set of non-crossing trees (Lemma~\ref{lem: terminal_cycle}) with leaves on the $\epsilon$-short north and south brick boundaries (Lemma~\ref{lem: terminal_NS}), with subgraphs guaranteed by Lemma~\ref{lem: newmap}.  The resulting solution is $\OPT_R$.

Let ${\cal P}_{xy}$ be a set of 2 (3) vertex disjoint paths in $\OPT_F$ for terminals $x,y$ requiring bi- (tri-)connectivity.  Each path in ${\cal P}_{xy}$ is broken into a sequence of small paths, each of which is either entirely in the interior of a brick or entirely in the mortar graph.  Let ${\cal P}_{xy}'$ be the set of these path sequences.

We define a map $\hat{\phi}$ for the paths in ${\cal P}_{xy}'$.  For a path $P$ of $\mathcal{P}'_{xy}$, if $P$ is on mortar graph, we define $\hat{\phi}(P)=P$; otherwise we define $\hat{\phi}(P) = \phi(T)$ where $T$ is the tree in OPT$_M$ containing $P$. Since $\phi(T)$ spans all leaves of $T$, $\hat{\phi}(P)$ also spans leaves of $T$ and connects endpoints of $P$. 

Let $P_1$ and $P_2$ be any two paths from distinct path-sequences inside of the same brick $B$. By Lemma~\ref{lem: final}, $P_1$ and $P_2$ can not belong to the same component of OPT$_M \cap \inter(B)$. So if $\hat{\phi}(P_1) \neq \hat{\phi}(P_2)$, then $\hat{\phi}(P_1)$ and $\hat{\phi}(P_2)$ are edge disjoint by the constructions of $\phi$ and $\hat{\phi}$. 
Otherwise, we know the image is a subgraph $\widehat{C}$ which guarantees triconnectivity for all vertices of $\widehat{C} \cap \partial B$ by Lemma~\ref{lem: newmap}. 
However,  there may be more than one paths from any path-sequence whose image is $\widehat{C}$, and the new $x$-to-$y$ paths may not be edge-disjoint in $\widehat{C}$. For this situation, we could shortcut the paths in $\widehat{C}$ such that each new $x$-to-$y$ path only contain one subpath in $\widehat{C}$.
Since there are at most three such subpaths and there endpoints are in $\partial B$,
$\widehat{C}$ contains edge-disjoint paths connecting the endpoints of those subpaths  by the last property of Lemma~\ref{lem: newmap}.

Therefore, the restructure step maintains that if terminals were 1-, 2- or 3-edge connected in $\OPT_F$, then they still are in $\OPT_R$.  We also see that, by construction of Lemma~\ref{lem: newmap}, the intersection $\OPT_R$ with the interior of a brick $B$ is a set of trees with leaves on $\partial B$.  Since the  Redirect step will only add edges of $\partial B$, this property does not change, proving one of the guarantees of the Structure Theorem.  

Further, the number of joining vertices is guaranteed by Lemma~\ref{lem: newmap} and the construction that is used for 2-ECP.  The number of joining vertices is on the same order as for 2-ECP, which depends only on $\epsilon$ as required.

\subsubsection{Redirect} For every joining vertex $j$ of $\OPT_R \cap B$ with $\partial B$ for a brick $B$, we add the path from $j$ to the nearest portal $p$ on $\partial B$.  This guarantees that the trees guaranteed by the Restructure step have leaves that are portals: this allows us to efficiently enumerate all possible Steiner trees in bricks whose terminals are portals to compute the spanner graph.

\subsubsection{Analysis of weight increase}
The analysis of the weight increase is exactly the same as for 2-ECP by Borradaile and Klein; the only difference are the weight in Lemma~\ref{lem: newmap} which is on the order of the weight in the equivalent Lemma used in 2-ECP.

This completes the proof of the Structure Theorem.

\section{Dynamic programming for $k$-ECP on graphs with bounded branchwidth} \label{sec:dp}

In this section, we give a dynamic program to compute the optimal solution of $k$-ECP problem on graphs with bounded branchwidth. 
This is inspired by the work of Czumaj and Lingas~\cite{CL98, CL99}.
Note that such graphs need not be planar. 
This can be used in our PTAS after the contraction step of the framework.

A {\em branch decomposition} of a graph $G=(V(G), E(G))$ is a hierarchical clustering of $E(G)$. It can be represented by a binary tree, called the {\em decomposition tree}, the leaves of which are in bijection with the edges of $G$. 
After deleting an edge $e$ of this decomposition tree, $E(G)$ is partitioned into two parts $E_1$ and $E_2$ according to the edges mapped to the leaves of the two subtrees. 
All the vertices common to $E_1$ and $E_2$ comprise the {\em separator} corresponding to $e$ in the decomposition.
The {\em width} of the decomposition is the maximum size of the separator in that decomposition.
The {\em branchwidth} of $G$ is the minimum width of any branch decomposition of $G$.

Let $G= (V(G), E(G), r)$ be an instance of $k$-ECP. Then $r\in \{0,1,\cdots,k\}$. We call a vertex a terminal if its requirement is positive. We first augment $G$ such that each edge becomes $k$ parallel edges. Our dynamic programming will work on this new graph $G$.
Given a branch decomposition of $G$, root the decomposition tree $T$ at an arbitrary leaf.  
For any node $q$ in $T$, let $L$ be the separator corresponding to its parent edge, and $E_1$ be the subset of $E(G)$ mapped to the leaves in the subtree rooted at $q$. 
Let $H$ be a subgraph of $G[E_1]$ such that it contains all terminals in $G[E_1]$. 
An {\em separator completion} of $L$ is a multiset of edges between vertices of $L$, each of which may appear up to $k$ times.  
A {\em configuration} of a terminal $v$ of $H$ in $L$ is a tuple $(A, B, r(v))$, where 
$A$ is a tuple $(a_1, a_2, \dots, a_{|L|})$, representing that there are $a_i$ edge-disjoint paths from $v$ to the $i$th vertex of $L$ in $H$, and $B$ is a set of tuples $(x_i, y_i, b_i)$, representing that there are $b_i$ edge-disjoint paths between the vertices $x_i$ and $y_i$ of $L$ in $H$. All the $\sum_{i=1}^{|L|} a_i+\sum_{i}b_i$ paths in a configuration are mutually edge-disjoint in $H$. 
We adapt a definition of Czumaj and Lingas~\cite{CL98, CL99}:

\begin{definition}
For any pair of terminals $u$ and $v$ in $H$, let $Com_H(u,v)$ be the set of separator completions of $H$ each of which augments $H$ to a graph where $u$ and $v$ satisfy the edge-connectivity requirement. 
For each terminal $v$ in $H$, let $Path_H(v)$ be a set of configurations of $v$ on $L$.
Let $Path_H$ be the set of all the non-empty $B$ in which all tuples can be satisfied in $H$.
Let $C_H$ be the set consisting of one value in each $Com_H(u,v)$ for all pairs of terminals $u$ and $v$ in $H$, and $P_H$ be the set consisting of one value in each $Path_H(v)$ for all terminal $v$ in $H$. 
We call the tuple $(C_H, P_H, Path_H)$ the {\em connectivity characteristic} of $H$, and denote it by $Char(H)$.
\end{definition}

Let $w$ be the width of the decomposition. 
Then $|L| \le w$.
Note that $H$ may correspond to multiple $C_H$ and $P_H$, so $H$ may have multiple connectivity characteristics. Further, each value in $P_H$ represents at least one terminal.  
For any $L$, there are at most $k^{O(w^2)}$ distinct separator completions ($O(w^2)$ pairs of vertices, each of which can be connected by at most $k$ parallel edges) and at most $2^{k^{O(w^2)}}$ distinct sets $C_H$ of separator completions. 
For any $L$, there are at most $k^{O(w^2)}$ different configurations for any terminal in $H$ since the number of different sets $A$ is at most $k^w$, the number of different sets $B$ is at most $k^{O(w^2)}$ (the same as the number of separator completions) and $k$ different choices for $r(v)$. 
So there are at most $2^{k^{O(w^2)}}$ different sets of configurations $P_H$, and at most $2^{k^{O(w^2)}}$ different sets $B$. Therefore, there are at most $2^{k^{O(w^2)}}$ distinct connectivity characteristics for a fixed $L$. 

\begin{definition}
A configuration of $v$ on $L$ is {\em connecting} if for any terminal $u$ in $V(G)\setminus V(H)$ the inequality $\sum_{i=1}^{|L|} a_i \geq \min\{r(v), r(u)\}$ holds where $a_i$ is the $i$th coordinate in $A$. That is, there are enough edge-disjoint paths from $v$ to the separator which can connect $u$ and $v$. $Char(H)$ is connecting if all configurations in its $P_H$ set are connecting. $H$ is connecting if at least one of $Char(H)$ is connecting. In the following, we only consider connecting connectivity characteristics and subgraphs.
\end{definition}

In the following, we need as a subroutine an algorithm to solve the following problem: when given a set of demands $(x_i, y_i, b_i)$ and a multigraph, we want to decide if there exist $b_i$ edge-disjoint paths between vertices $x_i$ and $y_i$ in the graph and all the $\sum_i b_i$ paths are mutually edge-disjoint. 
Although we do not have a polynomial time algorithm for this problem, we only need to solve this on graphs with $O(w)$ vertices, $O(kw^2)$ edges and $O(w^2)$ demands.
So even an exponential time algorithm is acceptable for our purpose here. Let $ALG$ be an algorithm for this problem, whose running time is bounded by a function $f(k,w)$, which may be exponential in both $k$ and $w$.

For a node $p$ of degree three in the decomposition tree $T$, let $q_1$ and $q_2$ be its two children and $q$ be its parent. 
Let $T_i$ be the subtree of $T$ rooted at $q_i$, let
$E_i$ be the subset of $E(G)$ corresponding to $T_i$ and let $L_i$ be the separator corresponding to the edge $pq_i$ for $i=1,2$. 
Let $L$ be the separator corresponding to $pq$.
For $i=1,2$, let $H_i$ be a subgraph of $G[E_i]$ that contains all the terminals. Let $H = H_1 \cup H_2$. Then we have the following lemma.
\begin{lemma}\label{lem: dp}
For any pair of $Char(H_1)$ and $Char(H_2)$, all the possible $Char(H)$ that could be obtained from $Char(H_1)$ and $Char(H_2)$ can be computed in 
$O(k^{w^2}f(k,w)+k^{w^2k^{w^2}})$
time.
\end{lemma}
\begin{proof}
We compute all the possible sets for the three components of $Char(H)$.

\noindent{\bf Compute all possible $\mathbf{C_{H}}$} 
$C_{H}$ contains two parts: the first part covers all pairs of terminals in the same $H_i$ for $i=1,2$ and the second part covers all pairs of terminals from distinct subgraphs. 

For the first part, we generalize each value $C \in C_{H_i}$ for $i=1,2$ into a possible set $X_C$. Notice that each separator completion can be represented by a set of demands $(x, y, b)$. 
For a candidate separator completion $C'$ on $L$, we combine $C'$ with each $B \in Path_{H_{3-i}}$ to construct a graph $H'$ and define the demand set the same as $C$. 
By running $ALG$ on this instance, we can check if $C'$ is a legal generalization for $C$. 
This may be computed in $k^{O(w^2)}w^2 + k^{O(w^2)}f(k,w)$ time for each $C$.
All the legal generalizations for $C$ form $X_C$. 

Now we compute the second part. For any pair of configurations $(A^1, B^1, r(u)) \in P_{H_1}$ and $(A^2, B^2, r(v)) \in P_{H_2}$ for $u\in H_1$ and $v\in H_2$, we compute possible $Com_H(u,v)$. Let $L' = L_1 \cap L_2$. We first count how many edge-disjoint paths between $u$ and $v$ could go through $L'$ by checking $A^1$ and $A^2$, and then check if a candidate separator completion $C'$ on $L$ can provide the remaining paths. All those $C'$ that are capable of providing enough paths form $Com_H(u,v)$. This can be computed in $w^2k^{O(w^2)}$ time for each pair of values.

A possible $C_H$ consists of each value in $X_C$ for every $C \in C_{H_i}$ for $i =1,2$ and each value in $Com_H(u,v)$ for all pairs of configurations of $P_{H_1}$ and $P_{H_2}$. To compute all the sets, we need at most $k^{O(w^2)}w^2 + k^{O(w^2)}f(k,w)$ time. There are at most $k^{O(w^2)}$ sets and each may contain at most $k^{O(w^2)}$ values. Therefore, to generate all the possible $C_H$ from those sets, we need at most $k^{w^2k^{O(w^2)}}$ time.

\noindent{\bf Compute all possible $\mathbf{P_H}$}
We generalize each configuration $(A, B, r(v))$ of $v$ in $P_{H_i}$ into a possible set $Y_v$. 
For each set $B'$ in $Path_{H_{3-i}}$, we construct a graph $H'$ by $A$, $B$ and $B'$ on vertex set $L_1 \cup L_2 \cup \{v\}$: if there are $b$ disjoint paths between a pair of vertices represented in $A$, $B$ or $B'$, we add $b$ parallel edges between the same pair of vertices in $H'$, taking $O(w^2)$ time.
For a candidate value $(A^*, B^*, r(v))$ corresponding to $L$, we define a set of demands according to $A^*$ and $B^*$ and run $ALG$ on all the possible $H'$ we construct for sets in $Path_{H_{3-i}}$. If there exists one such graph that satisfies all the demands, then we add this candidate value into $Y_v$. 
We can therefore compute each set $Y_v$ in $k^{O(w^2)}w^2 + k^{O(w^2)}f(k,w)$ time.
A possible $P_H$ consists of each value in $Y_v$. There are at most $k^{O(w^2)}$ such sets and each may contain at most $k^{O(w^2)}$ values. So we can generate all possible $P_H$ from those sets in $k^{w^2k^{O(w^2)}}$ time. 

\noindent{\bf Compute $\mathbf{Path_H}$}
For each pair of $B^1 \in Path_{H_1}$ and $B^2 \in Path_{H_2}$, we construct a graph $H'$ on vertex set $L_1 \cup L_2$: if two vertices are connected by $b$ disjoint paths, we add $b$ parallel edges between those vertices in $H'$. Since each candidate $B'$ on $L$ can be represented by a set of demands, we only need to run $ALG$ on all possible $H'$ to check if $B'$ can be satisfied.
We add all satisfied candidates $B'$ into $Path_H$. This can be computed in $k^{O(w^2)}w^2 + k^{O(w^2)}f(k,w)$ time.

Therefore, the total running time is $O(k^{w^2}f(k,w)+k^{w^2k^{w^2}})$.
For each component we enumerate all possible cases, and the correctness follows.
\end{proof}

Our dynamic programming is guided by the decomposition tree $T$ from leaves to root. 
For any node $q$ in $T$, let $T_q$ be the subtree of $T$ rooted at $q$ and $L_q$ be the separator corresponding $q$'s parent edge. Let $E_q$ be the subset of $E(G)$ corresponding to $T_q$.  
For each node $q$, our dynamic programming table is indexed by all the possible connectivity characteristics on the corresponding separator $L_q$. 
Each entry indexed by the connectivity characteristic $Char$ in the table is the weight of the minimum-weight subgraph of $G[E_q]$ that contains all the terminals in $G[E_q]$ and has $Char$ as its connectivity characteristic. 

\paragraph*{Base case} For each leaf of $T$, the only subgraph $H$ is the edge $uv$ contained in the leaf and the separator only contains its endpoints $u$ and $v$. There are three cases.
\begin{enumerate}
\item Both $u$ and $v$ are not terminals. $Com_H(u,v)$ contains all subsets of the multiset of edge $uv$ (up to $k$ times), including the empty set. $P_H$ is empty since there is no terminal. $Path_H$ contains one set: $\{(u, v, 1)\}$.

\item Only one of $u$ and $v$ is a terminal. W.l.o.g.~assume $u$ is the terminal. $Com_H(u,v)$ contains all subsets of the multiset of edge $uv$ (up to $k$ times), including the empty set. $Path_H(u)$ contains two configurations: $((k, 0), \{(u, v, 1)\}, r(u))$ and $( (k, 1), \emptyset, r(u) )$. $Path_H$ contains one set: $\{(u, v, 1)\}$.

\item Both $u$ and $v$ are terminals. $Com_H(u,v)$ contains the multisets of edge $uv$ that appears at least $\min\{ r(u), r(v) \} - 1$ times. $Path_H(u)$ contains two configurations: $((k,0), \{(u, v, 1)\}, r(u))$ and $((k,1), \emptyset, r(u))$, and $Path_H(v)$ contains two configurations: $((0, k), \{(u, v, 1)\}, r(v))$ and $((1,k), \emptyset, r(v))$. $Path_H$ contains one set: $\{(u, v, 1)\}$.
\end{enumerate}

For each non-leaf node $q$ in $T$, we combine every pair of connectivity characteristics from its two children to fill in the dynamic programming table for $q$. The root can be seen as a base case, and we can combine it with the computed results. 
The final result will be the entry indexed by $( \emptyset, \emptyset, \emptyset)$ in the table of the root.
If $E(G) = km$, then the size of the decomposition tree $T$ is $O(km)$. 
By Lemma~\ref{lem: dp}, we need $O(k^{w^2}f(k,w) + k^{w^2k^{w^2}})$ time to combine each pair of connectivity characteristics. 
Since there are at most $2^{k^{O(w^2)}}$ connectivity characteristics for each node, the total time will be $O(2^{k^{w^2}}k^{w^2}f(k,w)m + 2^{k^{w^2}}k^{w^2k^{w^2}}m)$.

\paragraph*{Correctness}
The separator completions guarantee the connectivity for the terminals in $H$, and the connecting configurations enumerate all the possible ways to connect terminals in $H$ and terminals of $V(G)\setminus V(H)$. So the connectivity requirement is satisfied . The correctness of the procedure follows from Lemma~\ref{lem: dp}. 

\subparagraph*{Acknowledgements}
We thank Hung Le, Amir Nayyeri and David Pritchard for helpful discussions.
\fi
\newpage
\bibliography{zerr.bib}
\bibliographystyle{plain}

\end{document}